\documentclass[runningheads,envcountsect]{llncs}
\usepackage[T1]{fontenc}
\usepackage{graphicx}
\usepackage[dvipsnames]{xcolor}
\usepackage{amsmath,amsfonts,amssymb}
\usepackage[english]{babel}
\usepackage{graphicx} 
\usepackage{mathabx} 
\usepackage{mathtools}
\usepackage{enumerate}
\usepackage{cite}
\usepackage[normalem]{ulem}
\usepackage[hidelinks]{hyperref}
\usepackage{tikz}
\usepackage[title]{appendix}%
\usepackage{tkz-graph}
\usepackage{tikz-cd}
\usepackage{comment}
\tikzset{%
    symbol/.style={%
        ,draw=none
        ,every to/.append style={%
            edge node={node [sloped, allow upside down, auto=false]{$#1$}}}
    }
}

\usetikzlibrary{arrows, chains, calc, matrix, positioning, scopes, patterns, shapes, fit, decorations.pathreplacing, backgrounds}
\usetikzlibrary{arrows.meta}

\tikzset{%
  >={Latex[width=2mm,length=2mm]},
            base/.style = {rectangle, rounded corners, draw=black,
                           minimum width=2cm, minimum height=0.75cm,
                           text centered, minimum width=2cm, fill=orange!15}
}

\usepackage{multirow}

\usepackage{algorithm}
\usepackage[noend]{algpseudocode}

\DeclareGraphicsExtensions{.png,.pdf}

\DeclareMathOperator{\determ}{det}

\DeclareMathOperator{\Span}{Span}

\DeclareMathOperator{\prob}{Pr}

\DeclareMathOperator{\Vol}{Vol}
\DeclareMathOperator{\Lap}{Lap}
\DeclareMathOperator{\D}{D}
\DeclareMathOperator{\Adj}{adj}

\newcommand{\lwt}[1]{\mathrm{wt}_{L}\!\left(#1\right)}

\newcommand\norm[1]{\left\lVert#1\right\rVert}
\newcommand\abs[1]{\left|#1\right|}
\newcommand\floor[1]{\left\lfloor#1\right\rfloor}
\newcommand\ceil[1]{\left\lceil#1\right\rceil}

\newcommand{\F}{\mathbb{F}}
\newcommand{\N}{\mathbb{N}}
\newcommand{\R}{\mathbb{R}}

\newcommand{\Z}{\mathbb{Z}}

\newcommand{\bb}{\mathbf{b}}
\newcommand{\bm}{\mathbf{m}}
\newcommand{\be}{\mathbf{e}}
\newcommand{\bg}{\mathbf{g}}
\newcommand{\bs}{\mathbf{s}}
\newcommand{\bc}{\mathbf{c}}
\newcommand{\bq}{\mathbf{q}}
\newcommand{\br}{\mathbf{r}}
\newcommand{\bv}{\mathbf{v}}
\newcommand{\bw}{\mathbf{w}}
\newcommand{\bx}{\mathbf{x}}
\newcommand{\by}{\mathbf{y}}
\newcommand{\bz}{\mathbf{z}}

\newcommand{\bA}{\mathbf{A}}
\newcommand{\bB}{\mathbf{B}}
\newcommand{\bG}{\mathbf{G}}
\newcommand{\bH}{\mathbf{H}}
\newcommand{\bI}{\mathbf{I}}

\newcommand{\C}{\mathcal{C}}
\renewcommand{\L}{\mathcal{L}}

\newcommand{\supp}{\text{Supp}}

\newcommand{\ConstructionA}{\ensuremath{\mathsf{Construction~A}}}

\newcommand{\ConstructionAG}{\ensuremath{\mathsf{Construction~A_{\mathbf{G}}}}}
\newcommand{\LeeSDP}{\ensuremath{\mathsf{LeeSDP}}}
\newcommand{\LeeDP}{\ensuremath{\mathsf{LeeDP}}}
\newcommand{\BDD}{\ensuremath{\mathsf{BDD}}}
\newcommand{\SVP}{\ensuremath{\mathsf{uSVP}}}
\newcommand{\LA}{\mathcal{L}_{\mathsf{A}}}

\newcommand{\LAG}{\mathcal{L}_{\mathsf{A}_{\mathbf{G}}}}

\newtheorem{thm}{Theorem}[section]
\newtheorem{lem}[thm]{Lemma}
\newtheorem{prop}[thm]{Proposition}
\newtheorem{cor}[thm]{Corollary}
\newtheorem{ex}[thm]{Example}
\newtheorem{probl}[thm]{Problem}
\newtheorem{defi}[thm]{Definition}
\newtheorem{rem}[thm]{Remark}

\begin{document}

\title{Lattice-Based Vulnerabilities in Lee Metric Post-Quantum Cryptosystems}
%
%
\author{Anna-Lena Horlemann\inst{1}\orcidID{0000-0003-2685-2343} \and Karan Khathuria\inst{2}\orcidID{0000-0002-9886-2770} \and
Marc Newman\inst{1}\orcidID{0009-0000-2818-3492} 
 \and Amin Sakzad\inst{3}\orcidID{0000-0003-4569-3384}\and
Carlos Vela Cabello\inst{1}\orcidID{0000-0003-3362-8817}}
%
%
\institute{University of St.Gallen, St. Gallen, Switzerland \\
\email{\{anna-lena.horlemann, marc.newman, carlos.velacabello\}@unisg.ch}
\and
Quantinuum, Partnership House, Carlisle Place, London SW1P 1BX, United Kingdom\\
\email{karan.khathuria@quantinuum.com}
\and
Monash University, Clayton, Australia\\
\email{amin.sakzad@monash.edu}
}
\maketitle              

\begin{abstract}
\begin{sloppypar}
Post-quantum cryptography has gained attention due to the need for secure cryptographic systems in the face of quantum computing.
Code-based and lattice-based cryptography are two prominent approaches, both heavily studied within the NIST standardization project.
Code-based cryptography---most prominently exemplified by the McEliece cryptosystem---is based on the hardness of decoding random linear error-correcting codes.
Despite the McEliece cryptosystem having been unbroken for several decades, it suffers from large key sizes, which has led to exploring variants using metrics other than the Hamming metric, such as the Lee metric.
This alternative metric may allow for smaller key sizes, but requires further analysis for potential vulnerabilities to lattice-based attack techniques.
In this paper, we consider a generic Lee metric based McEliece type cryptosystem and evaluate its security against lattice-based attacks. 
\end{sloppypar}
\end{abstract}
\section{Introduction}

In response to the threat posed by quantum computing to traditional cryptographic systems, post-quantum cryptography has gained significant attention over the last few years. Among the various approaches within post-quantum cryptography, code-based and lattice-based cryptography are two of the most widely studied research directions and constitute a majority of the current proposals in the NIST standardization project.

Code-based cryptography is founded on the hardness of  decoding (random linear) error-correcting codes, a problem that remains intractable for both classical and quantum computers in its general form. This branch of cryptography has its roots in the McEliece cryptosystem, proposed in 1978, a system that remains unbroken to date and, therefore, promises high security guarantees. On the other hand, it suffers from the drawback of requiring large public key sizes. Consequently, one of the main research tasks in code-based cryptography is to establish variants of the McEliece cryptosystem with smaller keys. One way of doing so is to use decoding metrics other than the originally proposed Hamming metric, e.g., the rank or the Lee metric, of which the latter is the main topic of this work.

Lattice-based cryptography, on the other hand, relies on the difficulty of solving problems in (high-dimensional) lattices, such as the Shortest Vector Problem (SVP) and the Learning With Errors (LWE) problem. Lattice-based schemes offer several compelling advantages, including strong security proofs and practical efficiency. The seminal works of Ajtai and Dwork \cite{AjtaiD97} in the late 1990s laid the groundwork for this domain, leading to the development of numerous cryptographic protocols \cite{Regev09,LyubashevskyPR13} that are both theoretically sound and practically viable.

The use of the Lee metric in code-based cryptography was first suggested in \cite{horlemann2021information} and has since been studied from a coding-theoretic perspective in, e.g., \cite{weger2022hardness,bariffi2021properties,bariffi2022information,chailloux2021classical}. These results suggest that the generic Lee syndrome decoding problem is (much) harder than its Hamming metric counter-part, which would imply that smaller codes could be used in a code-based cryptosystem when using the Lee metric instead of the Hamming metric. This would, in turn, lead to a reduced public key size. 

To complement the coding-theoretic perspective, it is well-known that the Lee metric over modular integer rings is the analog of the $\ell_1$-norm over the integers. It is therefore important to analyze the security of a Lee metric code-based cryptosystem with respect to lattice techniques (in the $\ell_1$- or $\ell_2$-norm). Exactly this approach has recently been used in \cite{FuLeakage} to break the signature scheme FuLeeca \cite{ritterhoff2023fuleeca}, which was submitted to the NIST standardization project. The attack exploits several properties of those Lee metric codes which arise from the specific parameters that were suggested---in particular, a large modulus and small minimum distance of the error vector. 

In this paper, we first consider a McEliece type cryptosystem over the Lee metric and then study the attackability of such Lee metric code-based cryptosystems with lattice techniques more generally. For this, we will focus on  public key encryption schemes (and not on digital signatures). In particular, we will derive complexity reductions to and from several known lattice problems including the bounded distance decoding problem (\BDD), the Lee-distance decoding problem (\LeeDP), and the unique shortest vector problem (\SVP); as shown in Fig. 1.
\begin{figure}
\centering
\begin{tikzpicture}[node distance=1.5cm,
    every node/.style={fill=white}, align=center]

  \node (LeeDP)             [base]              {\LeeDP};
  \node (BDD)    [base, right of=LeeDP, xshift=3.2cm]    {\BDD};
  \node (SVP)     [base, right of=BDD, xshift=3.2cm]   {\SVP};
  
  \draw[->]     (BDD) to[out=-30,in=-150] node[text width=1.5cm]{Thm. \ref{Th:ReductionBDD_to_USVP}} (SVP);
  \draw[->]     (SVP) to[in=30,out=150] node[text width=1.5cm]{Thm. \ref{Th:ReductionUSVP_to_BDD}} (BDD);
  
  \draw[->]     (LeeDP) to[out=-30,in=-150] node[text width=1.4cm]{Thm. \ref{Th:ReductionLeeDP_to_BDD}} (BDD);
  \draw[->]     (BDD) to[in=30,out=150] node[text width=1.4cm]{Thm. \ref{Th:ReductionBDD_to_LeeDP}} (LeeDP);

\end{tikzpicture}  
\caption{Scheme of the reductions for full rank integer lattices in the $\ell_1$-norm.}
\label{fig:Reductions}
\end{figure}
We will then analyze and find the values and parameters for which lattice reduction algorithms could be applied to Lee metric codes embedded in a lattice and compare the marginal error distributions of the Lee metric, the Hamming metric, and the $\ell_1$- and the $\ell_2$-norms for both the Laplace and Gaussian distributions.  

The paper is structured as follows. Section 2 provides all the necessary preliminaries, definitions and results needed for the rest of the paper. It also includes the Lee-McEliece cryptosystem. In Section 3, we first establish the relation between the shortest vector (in $\ell_1$-norm) of a lattice constructed based on Construction A and the minimum Lee distance of its underlying code. We further establish a two way complexity reduction between $\LeeDP$, $\BDD$, and $\SVP$, see Theorems~3.3 and 3.7. In Section 4, we study when the techniques in the FuLeakage attack~\cite{FuLeakage} that were applied to FuLeeca~\cite{ritterhoff2023fuleeca} can and cannot be applied to the cryptosystem in Section 2. Finally, we  establish connections between the Lee metric and the Laplace distribution and use it to compare Laplace and discrete Gaussian distributions in terms of R\'enyi divergence.

\section{Preliminaries}

We denote by $\Z_q$ the ring of integers modulo $q$. We will switch between two different representations of the elements of $\Z_q$, namely the standard representation $\{0,1,2,\dots,q-1\}$, and the representation centered at zero $\{-\lfloor(q-1)/2\rfloor,\dots, 0, \dots, \lfloor q/2 \rfloor\}$. If not specified, we will use $\Z_q=\{0,1,2,\dots,q-1\}$.

For a convex set $S \subseteq \R^n$ that spans a $k$-dimensional subspace, we will denote the $k$-dimensional relative volume of $S$ by $\Vol_k(S)$, i.e., the volume of $S$ in the linear space spanned by $S$. Given a set $U$ of $k$ vectors over $\R$, we will denote the span of $U$ in $\R$ by $\Span_\R(U) := \left\{\sum_{i=1}^k x_i \mathbf{u}_i \mid \ x_i \in \R,  \mathbf{u}_i \in U\right\}$.

\begin{defi}
	Let $\bA$ be a $n \times n$ invertible matrix and define $M_{i,j}$ to be the determinant of the $(n - 1) \times (n - 1)$ matrix obtained by removing the $i$th row and $j$th column from $\bA$.
	Then the \emph{adjugate} of $\bA$ is defined to be
	\begin{align*}
		\Adj(\bA) := {\left[ {(-1)}^{i + j} M_{j,i} \right]}_{1 \leq i, j \leq n}.
	\end{align*}
\end{defi}
	It is a well-known property of the adjugate that 
	\begin{align*}
		\Adj(\bA) \cdot \bA = \determ(\bA) \cdot \bI_n = \bA \cdot \Adj(\bA).
	\end{align*}

\subsection{Lee metric codes and the Lee-McEliece system}

\begin{defi}
For $x \in \Z_q$ we define the \emph{Lee weight} to be
\begin{equation*}
\lwt{x} := \min\{ \mid x \mid ,  \mid q-x  \mid \},
\end{equation*}
Then, for  $\bx \in \Z_q^n$, we define the \emph{Lee weight} to be the sum of the Lee weights of its coordinates,
\begin{equation*}
\lwt{\bx}:=  \sum_{i=1}^n \lwt{x_i}.
\end{equation*}
We define the \emph{Lee distance} of $\mathbf{x}$,  $\mathbf{y} \in \Z_q^n$ as 
$$ \text{d}_L(\mathbf{x}, \mathbf{y} ) := \lwt{\mathbf{x}- \mathbf{y}}.$$ 
\end{defi}

Note that for $q=2,3$ the Lee weight is equal to the Hamming weight $\mathrm{wt}_H$ in $\Z_q^n$, which is defined as
$$ \mathrm{wt}_H (\bx) := |\{i \mid x_i \neq 0\}| \quad \mbox{for all}\ \bx \in \Z_q^n.$$

Both for practical and theoretical reasons the following marginal distributions per coordinate of a vector with constant Lee---respectively, Hamming---weight $t$ will be useful.

\begin{lem}\label{lem:marginal_Lee}
\begin{enumerate}[(a)]
    \item \cite[Lemma 1]{bariffi2021properties} 
    Let $\bx \in \Z_q^n$ be a uniformly random vector with $\lwt{\bx}= t = T n$ for some $T \in [0,\floor{q/2}]$ such that $t \in \Z$. Further, let $E$ be the random variable representing a coordinate of $\bx$. Then, as $n$ tends to infinity, for any $j \in \Z_q$, 
    \begin{equation}
        F_T(j):=\prob(E = j) = \frac{\exp(-\beta \lwt{j})}{\sum_{i=0}^{q-1}\exp(-\beta \lwt{i})} ,
    \end{equation} 
where $\beta$ is the unique real solution to the constraint 
\begin{align}\label{beta}
    T =  \sum_{j=0}^{q-1} \lwt{j}  \frac{\exp(-x \lwt{j})}{\sum_{i=0}^{q-1}\exp(-x \lwt{i})} .
\end{align}

\item 
    Let $\bx \in \Z_q^n$ be a uniformly random vector with $wt_H(\bx)= t = \delta n$ for some $\delta \in [0,1]$ such that $t \in \Z$. Further, let $E$ be the random variable representing a coordinate of $\bx$. Then, as $n$ tends to infinity, for any $j \in \Z_q$, 
    $$H_\delta(j):=\prob(E = j) = \begin{cases}
        1-\delta & \mbox{if } j=0 \\
        \frac{\delta}{q-1} & \mbox{otherwise}
    \end{cases} .$$

\end{enumerate}
\end{lem}

Note that, even though the Lee and Hamming marginal distribution is an asymptotic result for growing $n$, the de facto distribution for small $n$ only differs by something very small. Therefore, we will use the marginals from above in our analysis. Furthermore, we will use $T$ for the relative Lee distance and $\delta$ for the relative Hamming distance.

\begin{rem}
For $q=2,3$ we get 
$$\beta = \log \left(\frac{1-\delta}{\delta}(q-1)\right) $$
 above in the Lee distribution, and hence the Lee distribution equals the Hamming distribution. 
\end{rem}

\begin{defi} 
Let $q$ be a positive integer.
\begin{enumerate}
    \item A \emph{code} over $\Z_q$ of length $n$ is a subset of $\Z_q^n$. 
    \item A \emph{(ring-)linear code} over $\Z_q$ of length $n$ is a $\mathbb{Z}/q\mathbb{Z}$-submodule of $\Z_q^n$. 
    \item The \emph{minimum Lee} \emph{distance} $d_L(\mathcal C)$ of a code $\mathcal C\subseteq \Z_q^n$ is the minimum of all Lee distances of distinct codewords of $\mathcal C$:
$$d_L(\mathcal{C})= \min\{ d_L(\bx,\by) \mid \bx,\by \in \mathcal{C} \ \mbox{with} \ \bx \neq \by\}.$$
\end{enumerate}
\end{defi}


Linear codes can be completely represented through a generator or a parity-check matrix.
\begin{defi}
A matrix $\mathbf G$ is called a \emph{generator matrix} for a (ring-)linear code $\mathcal C$ if its row space corresponds to $\mathcal C$. In addition, we call a matrix $\mathbf{H}$ a  \emph{parity-check matrix} for $\mathcal{C}$ if its kernel corresponds to $\mathcal{C}.$
\end{defi}
Note that such generator and parity-check matrices are not unique. If $q$ is not prime, even the number of rows of such matrices is not unique.

The general security assumption of code-based cryptography is based on the hardness of the syndrome decoding problem (SDP).\footnote{De facto this is not true for the McEliece cryptosystem, since the codes used are not random. However, we will not go into detail about this issue in this paper.} The Lee metric version of this is as follows:

\begin{probl}[Lee syndrome decoding problem ($\LeeSDP_t$)]
Given a linear code $\C$ over $\Z_q$ of length $n$ with parity check matrix $\bH \in \Z_q^{(n-k)\times n}$, a syndrome $\bs \in \Z_q^{n-k}$ and a positive integer $t \in \N$, find $\be \in \Z_q^n$ such that $\lwt{\be} \leq t$ and $\be \bH^\top =\bs$, where $^\top$ denotes the transposition operation.
\end{probl}

Note that this problem is equivalent to the general decoding problem for linear codes:

\begin{probl}[Lee decoding problem ($\LeeDP_t$)]
Given a linear code $\C$ over $\Z_q$ of length $n$, a vector $\br \in \Z_q^n$ and a positive integer $t \in \N$, find $\bc \in \C$ such that $\lwt{\br-\bc} \leq t$. \label{prob:Lee-DP}
\end{probl}

It was shown in \cite{weger2022hardness} that for uniformly random instances the syndrome decoding problem is NP-complete for any additive weight function, which includes the Lee metric. It is therefore a cryptographically interesting computationally hard problem to be used in public key cryptosystems. Algorithm \ref{alg1} shows a general setup of a McEliece-type public key encryption scheme with Lee metric codes.

 \begin{algorithm}[ht!]
\caption{Lee-McEliece cryptosystem}\label{alg1}
\begin{flushleft}
\textbf{Secret key:} The generator matrix $\bG_{sec} \in \Z_q^{k \times n}$ of an efficiently decodable Lee metric code with error-correction capacity $w\in \mathbb{N}$, and a Lee-isometry $\varphi$.

\textbf{Public key:} The generator matrix $\bG_{pub}=\varphi(\bG_{sec}) \in \Z_q^{k \times n}$ and $w$.

\textbf{Encryption:} To encrypt the message $\bm\in \Z_q^k$ choose an error vector $\be$ of Lee weight $w$ uniformly at random and create the cipher
$$ \bc= \bm\bG_{pub} + \be .$$

\textbf{Decryption:}
Decode $$ \varphi^{-1}(\bc)$$ in the secret code to retrieve $\varphi^{-1}(\bm)\bG_{sec}$. Recover $\bm$ through linear algebra operations and application of $\varphi$.
\end{flushleft}
\end{algorithm}

\begin{rem}
The Lee isometries are generally not transitive on the sphere of vectors with a fixed Lee weight. To prevent partial information leakage about the error vector during the encryption, this should be considered when choosing the secret linear code and generator matrix. Furthermore, the isometry $\varphi$ could be replaced by a near-isometry (i.e., maps that possibly change the weight of the vector by at most some prescribed value $t$), and the error weight in the encryption should be chosen to be $w-t$, such that the receiver can still uniquely decrypt. It is not the topic of this paper to analyze this issue, however it will be of paramount importance when suggesting a specific instance of such a cryptosystem.
\end{rem}

There are two main types of attacks that need to be analyzed in this setting: key recovery attacks, where the attacker can recover the secret linear code and its efficient decoding algorithm; and message recovery attacks, where the intruder recovers the message $\bm$ from the ciphertext $\bc$ without recovering the secret key. In this paper we will focus on the latter, by using known lattice techniques to recover the message.

\subsection{Lattice theory}

We assume that the space $\R^n$ is equipped with the $\ell_1$-norm $\norm{\bv}_1 := \sum_{i=1}^n |v_i|$. 
Note that this differs from the classical approach, where the $\ell_2$-norm is used.

\begin{defi}
    The $\ell_1$ distance between two vectors $\bv,\bw \in \R^n$ is denoted by $$d_1(\bv,\bw) := \norm{\bv-\bw}_1 .$$
\end{defi}

\begin{defi}
    Given $m$ linearly independent vectors $\bb_1,\ldots,\bb_m \in \R^n$, the lattice generated by them is given by \[\L(\bb_1,\ldots,\bb_m) := \left\lbrace \sum_{i=1}^m x_i \bb_i : x_i \in \Z \right\rbrace.\]
\end{defi}

For a vector $\br \in \R^n$, the distance between $\br$ and $\L$ is given by $d_1(\br,\L) := \inf \left\{d_1(\br,\bv): \bv\in \L \right\}$.
The shortest vector of a lattice $\L$ is the vector in $\L$ having the smallest $\ell_1$-norm. The length of the shortest vector is denoted by $\lambda_1(\L)$, the length of the shortest lattice vector that is not a multiple of the shortest vector is denoted by $\lambda_2(\L)$.

We can now state the two lattice problems in the $\ell_1$-norm that are of interest for us:
\begin{sloppypar}
\begin{probl}[$\alpha$-Bounded distance decoding problem ($\BDD_\alpha$)]\label{prob:BDD}
Given an integer lattice $\L$ and a vector $\br \in \Z^n$ such that $d_1(\br,\L) < \alpha \lambda_1(\L)$, find $\bv \in \L$ such that $d_1(\bv,\br) < \alpha \lambda_1(\L)$.
\end{probl}
\end{sloppypar}

\begin{probl}[$\gamma$-unique shortest vector problem ($\SVP_\gamma$)]
Given an integer lattice $\L$ such that $\lambda_2(\L) > \gamma \lambda_1(\L)$, find a non-zero vector $\bv\in\L$ of length $\lambda_1(\L)$.
\end{probl}
The connection between those two problems has already been studied in \cite{lyubashevsky2009bounded} as follows:

\begin{thm}\cite[Theorem 1]{lyubashevsky2009bounded}
For any $\gamma \geq 1$, there is a polynomial time reduction from \BDD$_{1/2\gamma}$ to $\SVP_\gamma$.
\end{thm}

\begin{thm}\cite[Theorem 2]{lyubashevsky2009bounded}
For any polynomially bounded $\gamma(n)=n^{O(1)}$, there is a polynomial time reduction from  $\SVP_\gamma$ to \BDD$_{1/\gamma}$.   
\end{thm}

We remark that the results in \cite{lyubashevsky2009bounded} were proven for the $\ell_2$-norm over $\R^n$; however, it was noted that the same proofs will hold for any other $\ell_p$-norm as well. Without loss of generality, we also assume that the above results hold for integer lattices and target vectors. We can thus use the $\ell_1$-versions as follows:
\begin{thm}\label{Th:ReductionBDD_to_USVP}
    For any $\gamma\geq1$, there is a polynomial time reduction from $\BDD_{1/(2\gamma)}$ to $\SVP_\gamma$ over the $\ell_1$-norm.
\end{thm}

\begin{thm}\label{Th:ReductionUSVP_to_BDD}
    For any polynomially bounded $\gamma(n)=n^{O(1)}$, there is a polynomial time reduction from $\SVP_\gamma$ to $\BDD_{1/\gamma}$ over the $\ell_1$-norm.
\end{thm}








Lastly, for our results in Section \ref{sec:containment} we will make use of the following two known results.


\begin{thm}\cite{Vaa79}\label{thm:vaa79}
	Let $C_n = {[-\tfrac{1}{2}, \tfrac{1}{2}]}^n \subseteq \R^n$, i.e., the $n$-dimensional unit cube centered at the origin.
	Let $P_k \subseteq \R^n$ be any $k$-dimensional linear subspace.
	Then $\Vol_k(C_n \cap P_k) \geq 1$.
\end{thm}


\begin{thm}\label{thm:minkowski}
	Let $\L$ be a $k$-dimensional lattice in $\R^n$ and let $S \subseteq \Span_\R(\L)$ be a convex set symmetric about the origin (i.e., $\bx \in S$ implies $-\bx \in S$).
	Suppose that $\Vol_k(S) > m \cdot 2^k \cdot \determ(\L)$.
	Then there are $m$ different pairs of vectors $\pm \bz_1, \ldots, \pm \bz_m \in S \cap \L \setminus \{0\}$.
\end{thm}
The above theorem is an extension of Minkowski's convex body theorem. Since the standard form of Minkowski's theorem is for full-dimensional lattices and $m=1$, we provide the proof of this version in Appendix \ref{app:Minkowski} for completeness. Our proof is based on the proofs from \cite[Theorem 20-21]{mink-ucsd} and \cite[Theorem 5-6]{mink-epfl}. 

\subsection{Distributions}

Let $F$ be a probability distribution over the sample space $X$. Then, we denote the support of $F$ by $\supp(F) := \{x \in X \mid F(x) \neq 0\}$. Throughout the paper, we may interchangeably use the same symbol to denote both the probability distribution and its density function. 

We define a continuous Gaussian distribution over $\R$ by its density function $\D_{\R,\sigma}(x) = \frac{1}{\sigma \sqrt{2 \pi}} \exp(- x^2/\sigma^2)$ and over a lattice $\L \subseteq \R^n$ as follows: 

\begin{defi}[Discrete Gaussian] For a lattice $\L$, the discrete Gaussian distribution $\mathrm{D}_{\L,\sigma}$ is defined by the probability density function
$$\D_{\L,\sigma}(\bx) := \frac{\exp(-\|\bx\|_2^2/2\sigma^2)}{\sum_{\by\in\L}\exp(-\|\by\|_2^2/2\sigma^2)},$$
for every $\bx\in\L$.
\end{defi}

Similarly, we define a continuous Laplace distribution over $\R$ by its density function $\Lap_{\R,b}(x) = \frac{1}{2b} \exp(- \abs{x}/b)$ and over a lattice $\L \subseteq \R^n$ as follows: 

\begin{defi}[Discrete Laplace] For a lattice $\L$, the discrete Laplace distribution $\mathrm{Lap}_{\L,b}$ with $b>0$ is defined by its probability density function
$$\Lap_{\L,b} (\bx) := \frac{\frac{1}{2b} \exp(-\|\bx\|_1/b)}{\sum_{\by\in\L}\frac{1}{2b}\exp(-\|\by\|_1/b)},$$
for every $\bx\in\L$.
\end{defi}

Given the integer lattice $\Z^n$, it is easy to check that $\D_{\Z^n,\sigma} = \prod_{i=1}^n \D_{\Z,\sigma}$ and $\Lap_{\Z^n,b} = \prod_{i=1}^n \Lap_{\Z,b}$. 

We use the R\'enyi and Kullback-Leibler divergence to measure the closeness of two distributions. 

\begin{defi}
Let $F$ and $G$ be discrete probability distributions satisfying $\supp(F)\subseteq\supp(G)$. Then, 
\begin{enumerate}
    \item \textbf{(R\'enyi divergence)} for any $a \in (1,\infty]$, the R\'enyi divergence of order $a$ between $F$ and $G$ is given by:
    \[R_a(F || G) := \begin{cases}
       \left( \sum\limits_{x \in \supp(F)} \frac{F(x)^{a}}{G(x)^{a-1}} \right)^{\frac{1}{a-1}} & \mbox{for}\ a \in (1,\infty) \\
        \max\limits_{x\in \supp(F)} \frac{F(x)}{G(x)} & \mbox{for}\ a=\infty
    \end{cases}\] 

    \item \textbf{(Kullback-Leibler divergence)} the Kullback-Leibler (KL) divergence between $F$ and $G$ is given by 
     \[KL(F||G) := \sum_{x \in \supp(F)} F(x) \log\left(\frac{F(x)}{G(x)}\right)  \]
    
\end{enumerate}
\end{defi}

The definitions are extended in a natural way to continuous distributions using integrals instead of the summations. Note that we define R\'{e}nyi divergence without taking the logarithm, which is standard in lattice-based cryptography. Given this, we see that the Kullback-Leibler divergence is a logarithm of the limit of R\'enyi divergence of order $a$ as $a$ goes to 1, i.e., $$KL(F||G) = \log \left( \lim\limits_{a \to 1} R_a(F||G) \right).$$
See \cite{R-KL_divergence} for a proof. 
We will use the following properties of R\'enyi and Kullback-Leibler divergence. We again refer to \cite{R-KL_divergence} for proofs. 

\begin{lem} Let $F$ and $G$ be  probability distributions with $\supp(F) \subseteq \supp(G)$. Further, let $F^{(n)} = F \times \cdots \times F$ and $G^{(n)} = G \times \cdots \times G$ be the product of $n$ independent and identical copies of $F$ and, respectively, $G$. Then, 
\begin{enumerate}
    \item \textbf{Multiplicativity of R\'enyi divergence}:
    \[R_a\left(F^{(n)}||G^{(n)}\right) = \prod_{i=1}^n R_a(F||G).\]

    \item \textbf{Additivity of Kullback-Leibler divergence}: 
    \[KL\left(F^{(n)}||G^{(n)}\right) = \sum_{i=1}^n KL(F||G).\]
\end{enumerate} \label{lem:product_divergence}
\end{lem}


\section{Complexity Reductions of Lee Metric Decoding Problems}
In this section we show that for bounded error vectors the Lee metric decoding problem (Problem \ref{prob:Lee-DP}) over linear codes reduces to the bounded distance decoding problem (Problem \ref{prob:BDD}) over lattices in the $\ell_1$-norm, and vice versa. All the results from this section are also summarized in Fig.~\ref{fig:Reductions}.

In general, we can always associate a lattice to a given linear code. One of the most common approaches is known as $\ConstructionA$, which takes a linear code in $\Z_q^n$ and translates it over $\Z^n$ using the vectors from $q\Z^n$.  

\begin{defi}[\(\ConstructionA\)]\label{Def:constructionA}
    Let $\C$ be a linear code in $\Z_q^n$ and let $\bG$ be a $k \times n$ generator matrix of $\C$. Then the \(\ConstructionA\) lattice associated to $\C$ is given by:
    \[
    \LA(\C) = \{ \bc \in \Z^n : \bc = \bG^\top \bx \bmod{q} \ \mbox{for some} \ \bx \in \Z^k \}.
    \]
\end{defi}

It can be easily seen that $\LA(\C) = \C + q\Z^n$, and hence $\LA(\C)$ does not depend on the choice of the generator matrix $\bG$. If the code $\C$ is clear from the context, we will simply denote $\LA(\C)$ by $\LA$.

With the representation of $\Z_q^n$ centered around zero, i.e., $\Z_q^n=\{-\floor{(q-1)/2},$ $\ldots,0,\ldots, \floor{q/2}\}^n$, we obtain that the construction of the lattice $\LA(\C)$ preserves the metric structure on $\C$, i.e., the length of the shortest $\ell_1$-norm vector in $\LA(\C)$ relates to the minimum Lee distance of $\C$. 

\begin{prop}
    Let $\C$ be a linear code in $\Z_q^n$. Then the $\ell_1$-norm of the shortest vector in the $\ConstructionA$ lattice $\LA$ is given by
    \[\lambda_1(\LA) = \min\{q,d_L(\C)\},\]
where $d_L(\C)$ is the minimum Lee distance of $\C$.
\end{prop}
This proposition has previously appeared in \cite{antonio2011decoding} and \cite{RushSloane} without a proof. Thus, for completeness we give the proof below. 
\begin{proof}
     For simplicity we assume that $q$ is odd and let $M = \floor{q/2}$. For an even $q$, the proof would be similar with only minor changes in the representation of $\Z_q$. 
     
     As described earlier, we represent elements of $\Z_q$ in $\Z$ by $\{-M,\ldots,M\}$. Using this representation, we get a one-to-one correspondence between the codewords in $\C$ and the lattice points of $\LA(\C)$ inside the $n$-cube $\left[-M, M\right]^n$. Note that each codeword $\bc \in \C$ and its representative, say $\tilde{\bc}$, in $\LA$ satisfy $\lwt{\bc} = \norm{\tilde{\bc}}_1$. This implies that $\lambda_1(\LA) \leq d_L(\C)$.
     Moreover, since $(q,0,\ldots,0) \in \LA$, we get $\lambda_1(\LA) \leq q$, and hence $\lambda_1(\LA) \leq \min\{q,d_L(\C)\}$.

     Now, to show that $\lambda_1(\LA) \geq \min\{q,d_L(\C)\}$, it is enough to show  $\lambda_1(\LA) \geq q$ or $\lambda_1(\LA) \geq d_L(\C)$. Let $\bx \in \LA$ be a lattice point such that $\norm{\bx}_1 = \lambda_1(\LA)$. If $\bx \bmod q = {\bf 0}$, then $\lambda_1(\LA) \geq q$ as $q$ is the smallest $\ell_1$-norm for a non-zero point in $\Z^n$. Now, if $\bx \bmod q \neq {\bf 0}$, then $\bx \in \LA \cap [-M,M]^n$ because, if $|x_i| > M$ for any $i$, then by either subtracting or adding $q$ to $x_i$ one can obtain another lattice point with $\ell_1$-norm strictly smaller than $\norm{\bx}_1 = \lambda_1(\LA)$, which is a contradiction. Since $\bx \in \LA \cap [-M,M]^n$, we get a codeword in $\C$ that corresponds to $\bx$ and has Lee weight equal to $\norm{\bx}_1 = \lambda_1(\LA)$. This implies that $d_L(\C) \leq \lambda_1(\LA)$.

\end{proof}
We remark that a similar result for the  Hamming distance and the $\ell_2$-norm for $\ConstructionA$ lattices has been given in~\cite{SS10} (see Corollary 2 therein).

\begin{thm}\label{Th:ReductionLeeDP_to_BDD}
    Let $\C$ be a linear code over $\Z_q$ with minimum Lee distance $d_L(\C)$. Then, for any $t = \alpha \min\{q,d_L(\C)\} \in \Z$ for some $\alpha \in (0,1)$, there is a polynomial time reduction from $\LeeDP_t$ on $\C$ to $\BDD_\alpha$ in the $\ell_1$-metric on $\LA(\C)$.
\end{thm}
\begin{proof}
    We consider an instance of $\LeeDP_t$ on $\C$ with $\br$ being the vector in $\Z_q^n$ to be decoded. Note that we can write $\br = \bc + \be$, where  $\bc \in \C$ is the closest codeword to $\br$ and $\be \in \Z_q^n$ is the corresponding error vector. Let $\tilde{\br}, \tilde{\bc}, \tilde{\be} \in [-M,M]^n$ be the corresponding representatives of $\br, \bc, \be$, respectively, in $\LA \cap [-M,M]^n$.
    Since $\tilde{\br} - \tilde{\bc}=\tilde{\be} \mod q,$ we get that 
    $$ \tilde{\br} - \tilde{\bc}=\tilde{\be} + \bv q \bI_n$$
    for some $\bv \in \Z^n$ (in fact, it is $\bv \in \{-1,0,1\}^n$). Here $\bI_n$ denotes the identity matrix of order $n$. Then $\bar\bc :=\tilde{\bc} - \bv q \bI_n$ is an element of $\LA$ and fulfills $\tilde{\br} - \bar{\bc}=\tilde{\be}$. This implies $d_1(\tilde{\br},\LA) \leq \norm{\tilde{\br} - \bar{\bc}}_1 =\norm{\tilde{\be}}_1= \lwt{\be} \leq t = \alpha \lambda_1(\LA)$.
    Hence, we get an instance of $\BDD_\alpha$ for a received vector $\tilde{\br}$ with $d_1(\tilde{\br},\LA) \leq \alpha \lambda_1(\LA)$. The \(\BDD_\alpha\) oracle now gives a lattice vector $\bx$ satisfying   $d_1(\tilde{\br},\bx) \leq \alpha \lambda_1(\LA) = t$. Let $\bc_x:=\bx  \pmod{q}$, then we have that $\bc_x \in \C$ (according to the definition of \(\ConstructionA\) lattices) and $\lwt{\br-\bc_x}\leq t$.
\end{proof}

\begin{rem}
    In the case when we have a $\LeeDP_t$ instance with $t= \alpha d_L(\C)$ and  $d_L(\C) >q$, the reduction to $\BDD_\alpha$ does not hold. Note that in this case, we cannot directly apply the $\BDD_\alpha$ oracle, like we did in the proof of Theorem \ref{Th:ReductionLeeDP_to_BDD}, because we may not satisfy $d_1(\tilde{\br},\LA) \leq \alpha \lambda_1(\L_A) = \alpha q$ for any $\alpha \in (0,1)$. 
\end{rem}

\begin{sloppypar}
\begin{lem}\label{Thm:Decon}
	Let $\L \subseteq \Z^n$ be a full rank integer lattice with basis vectors $\{\bb_1, \ldots, \bb_n\}$, let
	\begin{align*}
		\bB = \begin{bmatrix}
		\bb_1 \\
		\vdots \\
		\bb_n
		\end{bmatrix},
	\end{align*}
	and let $q = \determ(\bB)$.
	Let $\C_{\mathsf{A}}(\bB) \subseteq \Z_q^n$ be the code generated by the vectors of $\bb_i \pmod{q}$.
	Then $\LA(\C_{\mathsf{A}}(\bB)) = \L$.
\end{lem}
\end{sloppypar}
\begin{proof}
	Let $\widecheck{\bb_i} \in \Z_q^n$ be the coordinate-wise reduction of $\bb_i$ modulo $q$ and let $\widetilde{\bb_i} \in \Z^n$ be $\widecheck{\bb_i}$ considered as an integer vector.
	For $1 \leq i \leq n$, let $\bq_i \in \Z^n$ be the vectors with all zeros except for a $q$ in the $i$th coordinate.
	Then, for all $i$,
	\begin{align*}
		\bb_i
		= \widetilde{\bb_i}
		+ \sum_{j = 1}^n c_{i,j} \bq_j
	\end{align*}
	for some integers $c_{i,j}$.
	By definition, $\LA(\C_{\mathsf{A}}(\bB))$ is generated by the vectors $\widetilde{\bb_1}, \ldots, \widetilde{\bb_n}, \bq_1, \ldots, \bq_n$ so $\bb_i \in \LA(\C_{\mathsf{A}}(\bB))$ for all $i$, and therefore $\L \subseteq \LA(\C_{\mathsf{A}}(\bB))$.
	Conversely, because $\Adj(\bB) \cdot \bB = q \bI_n$ (or, alternatively, see~\cite[Theorem 16]{mink-ucsd}), each $\bq_j \in \L$ and then, for all $i$, $\widetilde{\bb_i} = \bb_i - \sum_{j = 1}^n c_{i,j}  \bq_j \in \L$ so we have $\LA(\C_{\mathsf{A}}(\bB)) \subseteq \L$.
\end{proof}

\begin{thm}\label{Th:ReductionBDD_to_LeeDP}
	Let $\L \subseteq \Z^n$ be a full-rank integer lattice with basis $\bB$.
	Then for some $\alpha \in (0, 1)$ and $t = \alpha \lambda_1(\L)$ there exists a polynomial time reduction from $\BDD_\alpha$ with received vector $\br \in \Z^n$ in the $\ell_1$-norm on $\L$ to $\LeeDP_t$ on some $\C \subseteq \Z_q^n$, for some $q$.
\end{thm}
\begin{proof}
	Given $\br \in \Z^n$ such that $d_1(\br, \L) < \alpha \lambda_1(\L)$, we know there exists some $\bv \in \L$ such that $d_1(\bv, \br) < \alpha \lambda_1(\L)$.
	Let $\C \subseteq \Z_q^n$ be the code obtained (as above) by the reducing the lattice modulo $q = \determ(\bB)$ where we represent coordinates of the vectors in $\Z_q^n$ with integers between $\lceil -q / 2 \rceil$ and $\lfloor q / 2 \rfloor$ (omitting the value $\lceil - q / 2 \rceil$ in the case that $q$ is even).
	Let $\widecheck{\br} = \br \pmod{q}$ and $\widecheck{\bv} = \bv \pmod{q}$.
	Then $d_L(\widecheck{\bv}, \widecheck{\br}) \leq d_1(\bv, \br) < t$.

	Given input $\C$ and $\widecheck{\br}$, $\LeeDP_t$ outputs a codeword $\widecheck{\bc} \in \C$ such that $d_L(\widecheck{\bc}, \widecheck{\br}) < t$.
	Now, consider $\widecheck{\be} = \widecheck{\br} - \widecheck{\bc}$, and let $\widetilde{\be}, \widetilde{\br}, \widetilde{\bc} \in \Z^n$ be the vectors $\widecheck{\be}, \widecheck{\br}, \widecheck{\bc}$ considered as an integer vectors, and set $\bs = \br - \widetilde{\be}$.
	By Lemma~\ref{Thm:Decon}, we know that $\widetilde{\bs} = \widetilde{\br} - \widetilde{\be} \in \L$ and $\bq_1, \ldots, \bq_n \in \L$.
	Additionally, we know that there exist $c_1, \ldots, c_n \in \Z$ such that $\br = \widetilde{\br} + \sum_{i = 1}^n c_i \bq_i$.
	Then
	\begin{align*}
		\br - \widetilde{\be}
		= \widetilde{\br} + \sum_{i = 1}^n c_i \bq_i - \widetilde{\be}
		= \widetilde{\bs} + \sum_{i = 1}^n c_i \bq_i
		\in \L.
	\end{align*}
	Lastly, note that
	\begin{align*}
		\norm{\widetilde{\be}}_1
		= \sum_{i = 1}^n \abs{\widetilde{e_i}}
		= \sum_{i = 1}^n \min \{\abs{\widecheck{e_i}}, \abs{q - \widecheck{e_i}}\}
		= \lwt{\widecheck{\be}}
		< t
		= \alpha \lambda_1(\L).
	\end{align*}
	Thus $\bv = \br - \widetilde{\be}$ is a valid solution to $\BDD_\alpha$.
\end{proof}


\section{Containment of Finite Codes in Construction A Lattices}\label{sec:containment}

In this section we will study when a code over $\Z_q$ is completely contained in the lattice generated by a given generator matrix of the code. We remark that the containment was one of the crucial factors in the FuLeakage attack~\cite{FuLeakage} on the signature scheme FuLeeca\cite{ritterhoff2023fuleeca} since this allowed them to reduce the attack complexity by reducing the lattice dimension. This would also work as a message-recovery attack on a Lee-McEliece cryptosystem if the majority of the codewords are contained in a lower dimensional sublattice of \ConstructionA. We therefore analyze the cardinality of the intersection of the code with the lattice generated by a generator matrix of the code (which is always a sublattice of \ConstructionA \ and the \ConstructionA \  lattice consists of a union of affine shifts of the sublattice). 


We first introduce a fixed notation for the lattice generated by the generator matrix of the code:



\begin{defi}[\(\ConstructionAG\)]\label{Def:constructionAp}
	Let $\C$ be a linear code in $\Z_q^n$, let $\bG$ be a $k \times n$ generator matrix of $\C$, and let $\bg_i$ with $i \in \{1, \ldots, k\}$ be its rows.
	Then, the \(\ConstructionAG\) lattice associated to $\bG$ is given by:
	\[
		\LAG(\C) = \left\lbrace \sum_{i = 1}^k z_i \bg_i\,:\,z_i \in \Z \right\rbrace.
	\]
\end{defi}

In the FuLeeca attack in~\cite{FuLeakage} it was experimentally shown that the secret codewords (very short vectors) and signatures of the scheme in \ConstructionA \ are both always contained in the \ConstructionAG \ lattice.
However, this is not true in general and running BDD or SVP solvers on the whole \ConstructionA \ lattice is usually not feasible. Therefore, we would like to know when $\C$ is contained in $\LAG(\C)$, or---if not---how many elements of the code are contained in $\LAG(\C)$. 

We will use a generalized version of Minkowski's bound (see Theorem \ref{thm:minkowski}) to derive a lower bound on the cardinality of $\C \cap \LAG(\C)$:


\begin{thm}
\label{prop:lowerlag}
	Let $\C$ be a linear code in $\Z_q^n$, let $\bG$ be a generator matrix of $\C$ considered in $\Z^{k \times n}$, and let
    \begin{align*}
        M := \begin{cases}
            \frac{q - 1}{2} &\text{for $q$ odd} \\
            \frac{q}{2} - 1 &\text{for $q$ even}
        \end{cases}.
    \end{align*}
	Then
	\[\abs{\C \cap \LAG(\C)} \geq
		2 m + 1,
	\] where $m$ is the largest positive integer strictly less than $\frac{{(2M)}^k}{2^k \sqrt{\det(\bG \bG^\top)}}$.
\end{thm}

\begin{proof}
	Let $S := {[-M, M]}^n \cap \Span_\R(\LAG(\C))$ where $\Span_\R(\LAG(\C)) \subseteq \R^n$ is the $k$-dimensional $\R$-subspace spanned by $\LAG(\C)$.  Note that in the case where $q$ is even, due to the requirement that $S$ be symmetric, we are omitting some possible lattice points from our lower bound by excluding those whose coordinates take values of $q / 2$.
	By Theorem~\ref{thm:vaa79}, we have that $\Vol_k(S) \geq {(2M)}^k$.
 
	Let $m$ be the largest positive integer strictly less than $\frac{{(2M)}^k}{2^k \sqrt{\det(\bG \bG^\top)}}$.
	Then
	\[
		\frac{\Vol_k(S)}{2^k \determ(\LAG(\C))}
		= \frac{\Vol_k(S)}{2^k \sqrt{\determ(\bG \bG^\top)}}
		\geq \frac{{(2M)}^k}{2^k \sqrt{\determ(\bG \bG^\top)}}
		> m.
	\]
	By Theorem~\ref{thm:minkowski}, we know that there are then at least $m$ pairs of non-zero vectors in $S \cap \LAG(\C) \subseteq \C$.
	Including the zero vector gives us the bound.
\end{proof}

\begin{rem}
    Note that the lower bound introduced in Theorem \ref{prop:lowerlag} is inversely proportional to $\sqrt{\det(\bG\bG^\top)}$, i.e., it maximizes when $\sqrt{\det(\bG\bG^\top)}$ is minimal. It is well-known that $\sqrt{\det(\bG\bG^\top)}$ is minimized for unimodular even or Type II lattices (which are closely related with self-dual codes), see e.g.,~\cite{conwaylattices}.  This is an indication that self-dual codes might admit a large number of codewords in $\LAG$ and would hence be cryptographically insecure. Similary, the bound increases for growing $q$, indicating that a very large $q$ will likely be insecure.
\end{rem}




In the above theorem we establish a lower bound for the cardinality of the intersection of a code $\mathcal{C}$ over $\Z_q^n$. 
Naturally, we would also like to derive an upper bound on this number. A trivial upper bound is the cardinality of $\C$, that is, $\abs{\C \cap \LAG(\C)}\leq\abs{\C}$.
We remark that there exists a reverse Minkowski bound, which could be used to derive another upper bound---however, it turns out that doing so results in a bound above the trivial bound, which is not useful. 

In general, the lower bound derived in Theorem \ref{prop:lowerlag} and the trivial upper bound are not tight, however in special cases they are. In the following examples we illustrate this fact.

\begin{ex}\label{Example:LattCont_1}
    Let $\C_1$ and $\C_2$ be linear codes in $\Z_7^2$ with generator matrices $\bG_1=\begin{pmatrix}1 &1\end{pmatrix}$ and $\bG_2=\begin{pmatrix}1 &2\end{pmatrix}$ respectively.
    \begin{figure}
        \centering
    \begin{tikzpicture}[
  bluenode/.style={shape=circle, draw=blue, line width=2,scale=.5},
  greennode/.style={shape=circle, draw=OliveGreen, line width=2,scale=.5},
  rednode/.style={shape=circle, draw=red, line width=2,scale=.5}
  ]
    
    \node (f1) at (-3,0) {\includegraphics[scale = 0.35]{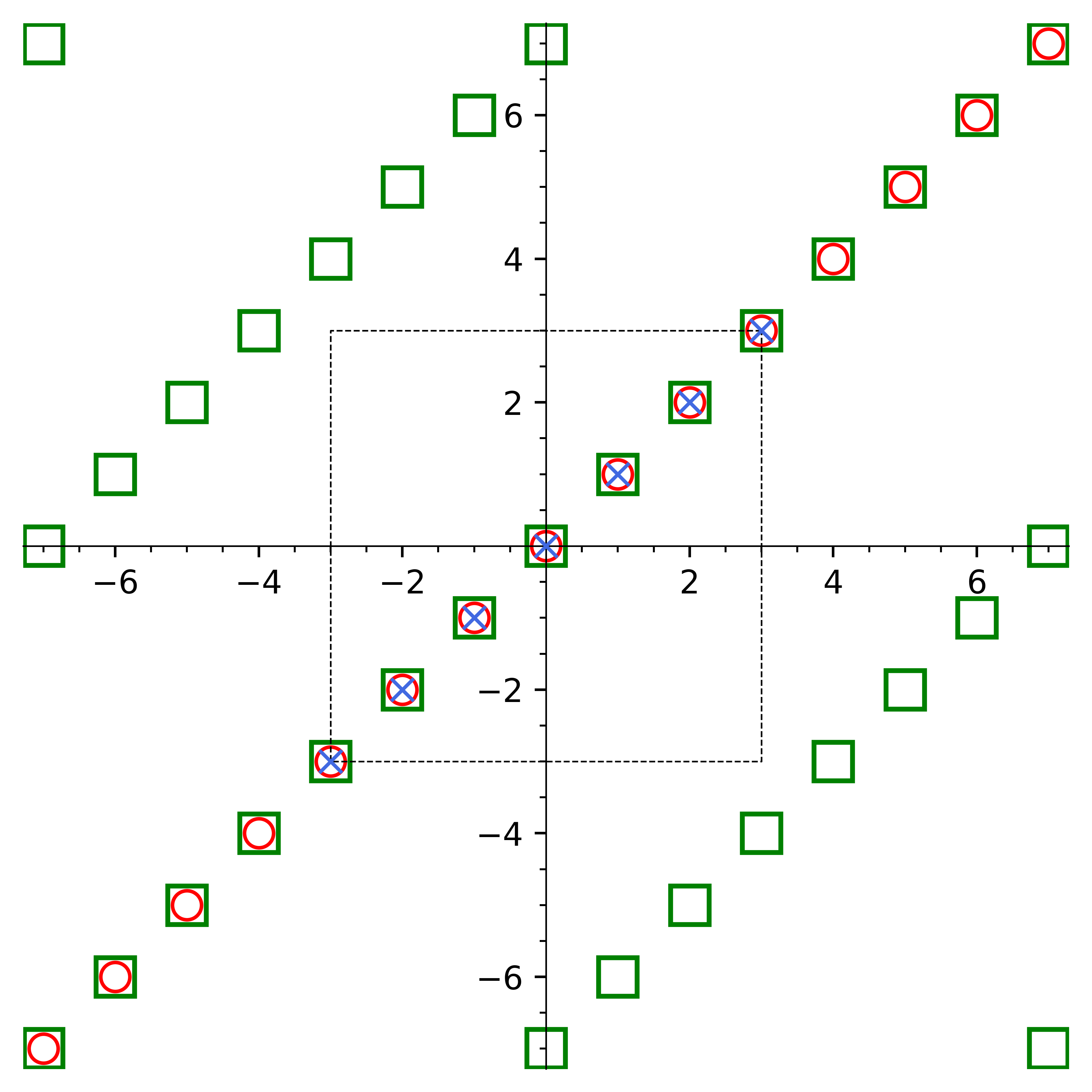}};
    \node[below of = f1, node distance=2.34cm] {$\C_1 = \langle (1,1) \rangle \in \Z_7^2$};

    \node at (0,1.4) {{\setlength{\fboxsep}{0pt}\fbox{\includegraphics[scale = 0.5]{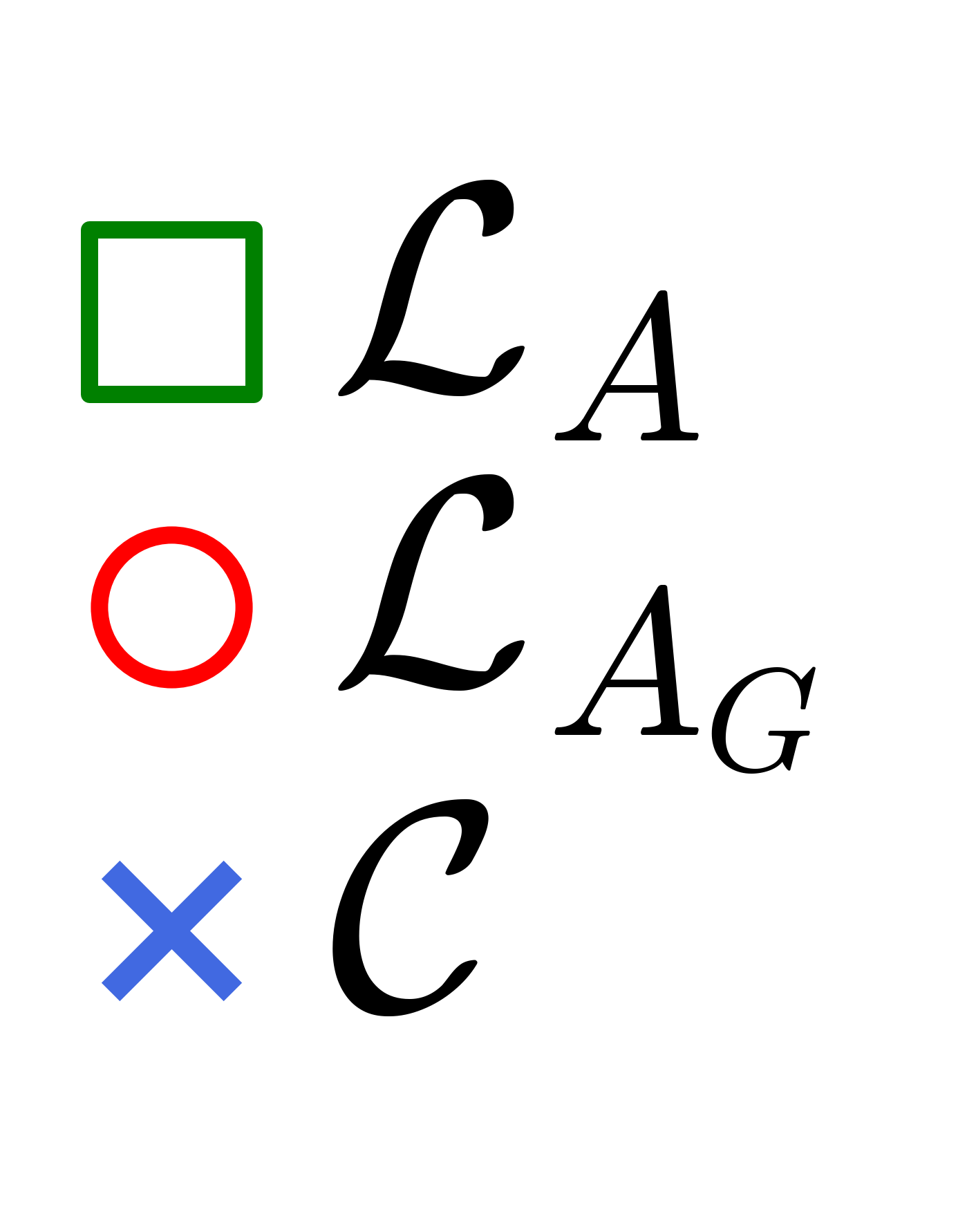}}}};


    \node (f2) at (3,0) {\includegraphics[scale = 0.35]{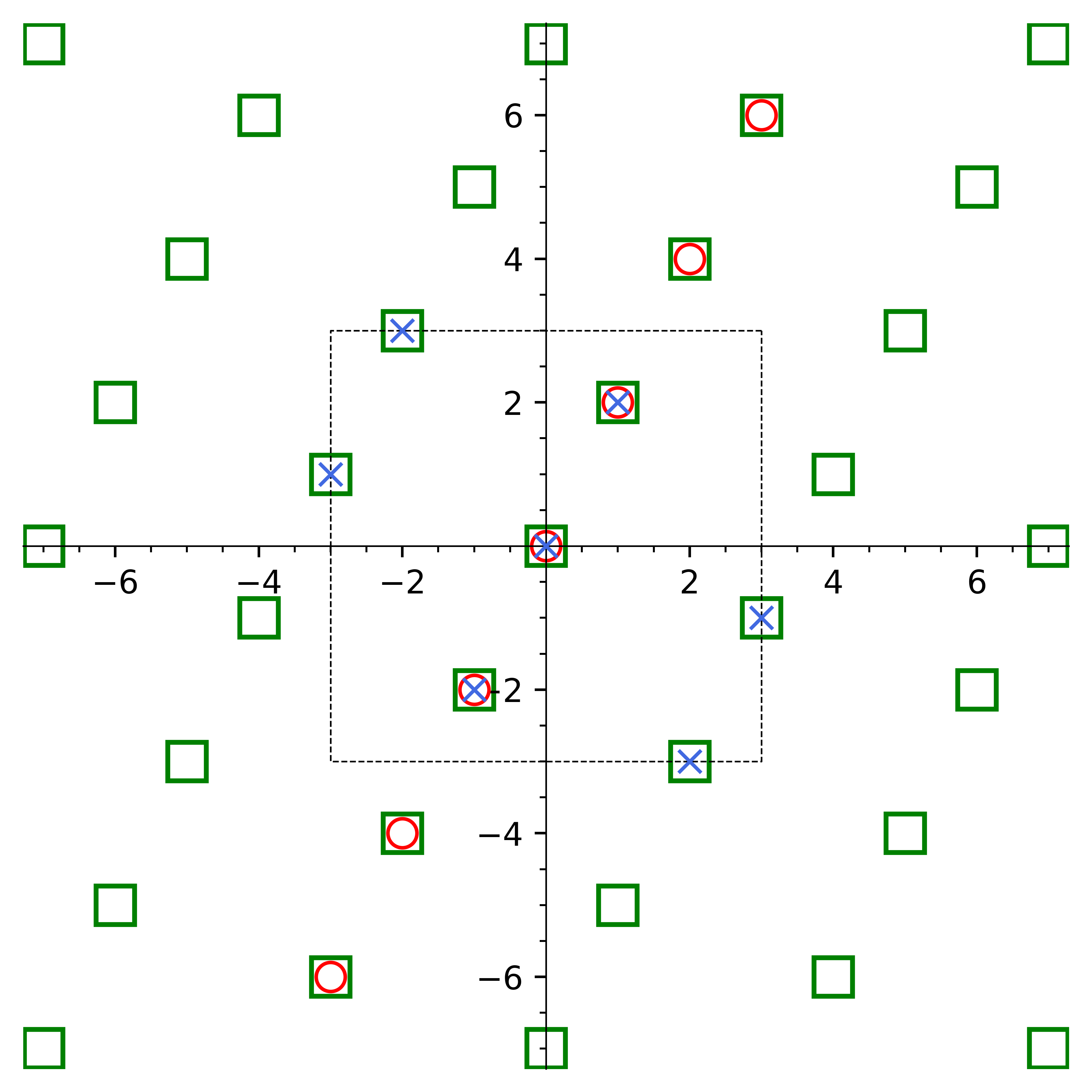}};
    \node[below of = f2, node distance=2.34cm] {$\C_2 = \langle (1,2) \rangle \in \Z_7^2$};
    \end{tikzpicture}
        \caption{Lattices $\LA$ and $\LAG$ for $\C_1$ and $\C_2$ in Example~\ref{Example:LattCont_1}.}\label{fig:Example:LattCont_1}
    \end{figure}
    Fig.~\ref{fig:Example:LattCont_1} depicts these two codes and their corresponding lattices $\LA$ and $\LAG$. Note that $\abs{\C_1 \cap \LAG(\C_1)}=7$ and $\abs{\C_2 \cap \LAG(\C_2)}=3$. Both codes have $7$ elements, which is also the trivial upper bound for $\abs{\C_i \cap \LAG(\C_i)}$, for $i=1,2$. We see that the trivial upper bound for $\abs{\C \cap \LAG(\C)}$ is attained for $\C_1$, but not for $\C_2$. Now, the lower bound from Theorem \ref{prop:lowerlag} for these two cases is
    \begin{align*}
        \abs{\C_1 \cap \LAG(\C_1)} &\geq 2\left\lfloor\frac{3}{\sqrt{2}}\right\rfloor+1=5, \\
        \abs{\C_2 \cap \LAG(\C_2)} &\geq 2\left\lfloor\frac{3}{\sqrt{5}}\right\rfloor+1=3
    \end{align*}
    respectively. We see that the lower bound for $\abs{\C \cap \LAG(\C)}$ is attained for $\C_2$, but not for $\C_1$.
\end{ex}

\begin{ex}\label{Example:LattCont_2}
Let $\C_3\subseteq \Z_7^3$ and $\C_4\subseteq \Z_{13}^5$ be linear codes with generator matrices $$\bG_3=\begin{pmatrix}3 & 1 & 2 \\ 3 & 2 & 3\end{pmatrix} \text{ and }\bG_4=\begin{pmatrix}3 & 1 & 2 &5 &-4\\ 3 & 2 & 3 &6 & -1\\-1&2&5&-5&6\end{pmatrix},$$ respectively. In these two cases, both the trivial upper bound and the lower bound from Theorem \ref{prop:lowerlag} are not tight since $\abs{\C_3 \cap \LAG(\C_3)}=19$ and $\abs{\C_4 \cap \LAG(\C_4)}=17$ and the bounds give
    \begin{align*}
        49=\abs{\C_3}&\geq\abs{\C_3 \cap \LAG(\C_3)} \geq 2\left\lfloor\frac{9}{\sqrt{19}}\right\rfloor+1=5,\\
         2197=\abs{\C_4}&\geq\abs{\C_4 \cap \LAG(\C_4)} \geq 2\left\lfloor\frac{6^3}{\sqrt{23804}}\right\rfloor+1=3.
    \end{align*}
    
\end{ex}

\begin{rem}
    Both the lower bound from Theorem \ref{prop:lowerlag}, and the actual number $\abs{\C \cap \LAG(\C)}$ generally depend on the choice of generator matrix $\bG$ and are not a code invariant (see Example \ref{Ex:GnotInvariant}). 
    It would be useful to find a characterization or tighter bounds to understand when a generator matrix leads to a big (or small) intersection number. This seems to be a complex task, since the number of zeros (i.e., the Hamming weight), the number of different Lee weights, and the largest Lee weight (i.e., the $\ell_\infty$-norm) of the basis vectors have an impact on the wrap-around behavior (at the boundaries) of the code over $\Z_q$, when represented over $\Z$.
\end{rem}


To give more insight into the wrap-around behavior we describe it in the one-dimensional case. However, already for codes of dimension $2$, the situation is much more complex and is left as an open problem for future work. 

\begin{prop}
 Let us represent $\Z_q$ centered at zero, and let $M = \floor{q/2}$.
Further, let $\C$ be a one-dimensional linear code on $\Z_q^n$ and $\bG$ a $1\times n$ generator matrix. Then we have that:
\begin{enumerate}
    \item if $\abs{\C \cap \LAG(\C)}=q$, then all non-zero entries of $\bG$ are $\pm 1$. This is the only case where $\C \subseteq \LAG(\C)$.
   \item if $||\bG||_{\infty}$, the largest magnitude among the entries of $\bG$, is equal to $t\in\{0,1,\dots,M\}$, then
   \[\abs{\C \cap \LAG(\C)}=\left\lbrace\begin{array}{ll}
       2\left\lfloor M/t\right\rfloor+1 & \text{ if }q \text{ odd}  \\[.3cm]
       \left\lfloor M/t\right\rfloor+\left\lfloor(M-1)/t\right\rfloor+1 & \text{ if }q \text{ even}  \\
   \end{array}\right. .\]
\end{enumerate}
\end{prop}
\begin{proof}
 We assume for simplicity that $q$ is odd (the even case is analogous). 
\begin{enumerate}
    \item
   Denote by $g_i$ the $i$-th entry of $\bG$. We can easily count all the non-zero integer multiples of an entry $g_i$ that are within $\Z_q$ (for $\lambda \in \Z\backslash\{0\}$): 
    $$-\frac{q-1}{2}\leq \lambda g_i \leq \frac{q-1}{2} \iff   |g_i| \leq \left|\frac{q-1}{2\lambda}\right|.$$
    Thus, for $q-1$ non-zero multiples of $\bG$ to be in $\Z_q^n$ we need, in particular for $\lambda=\pm(q-1)/2$, to have $|g_i| \leq \left|\frac{q-1}{2\lambda}\right|=1$, which implies that $g_i\in \{0,\pm 1\}$.
    \item 
     With a similar counting argument as above we get that for $\lambda \in \Z$, we have $|\lambda g_i| \leq|\lambda t| $ and
    $$|\lambda t| \leq \frac{q}{2} \iff |\lambda| \leq \frac{q}{2t} = \frac{M}{t}, $$
    i.e., exactly for $\lambda \in \{-\left\lfloor M/t\right\rfloor, \dots, \left\lfloor M/t\right\rfloor\}$ the $\lambda$-multiple of $\bG$ is contained in $\LAG(\C)\cap \Z_q^n$, which implies the statement.
\end{enumerate}
\end{proof}

\begin{ex}\label{Ex:GnotInvariant}
    Let $\C_5$ be a linear code in $\Z_{11}^2$ and $\bG_5=\begin{pmatrix}1 &2\end{pmatrix}$ and $\bG_5'=\begin{pmatrix}5 &-1\end{pmatrix}$ be two generator matrices. It is easy to see that $\abs{\C_5}=11$ but when looking at the cardinality of the intersection with the lattice we have that
    $$
    \abs{\C_5 \cap \mathcal{L}_{\mathsf{A}_{\mathbf{G_5}}}(\C_5)}=5\text{ and } \abs{\C_5 \cap \mathcal{L}_{\mathsf{A}_{\mathbf{G}_5'}}(\C_5)}=3.
    $$
    Now, let $\C_3\subseteq\Z^3_7$ be the same code as in Example \ref{Example:LattCont_2}, then $\C_3$ can also be generated with
     $$\bG_3'=\begin{pmatrix}0 & 1 & 1 \\ 3 & 0 & 1\end{pmatrix} \text{ and }\bG_3''=\begin{pmatrix}0 & 2 & 2 \\ 3 & 2 & 3 \end{pmatrix}.$$
     Again, when checking the cardinality of the intersection we obtain
     $$
     \abs{\C_3 \cap \mathcal{L}_{\mathsf{A}_{\mathbf{G_3'}}}(\C_3)}=20\text{ and } \abs{\C_3 \cap \mathcal{L}_{\mathsf{A}_{\mathbf{G}_3''}}(\C_3)}=9.
     $$
     This illustrates the dependency on the choice of generator matrix.
\end{ex}


\section{Comparison of Error Distributions}

In this section we will compare the different error distributions related to the metrics described before. First, we will compare the Hamming and the Lee metric. Then we will show the connection between the Lee metric and the Laplace distribution, which motivates us to compare the Laplace and the Gaussian distribution (to compare the behavior of the $\ell_1$- and $\ell_2$-norm). For the discrete distributions we will use R\'enyi divergence, whereas for the continuous distributions, we will use Kullback-Leibler convergence.

Let us recall from Lemma \ref{lem:marginal_Lee} that for a uniformly random vector $\bx \in \Z_q^n$ with normalized Lee distance $T$, the marginal Lee distribution is given by
    \begin{equation*}
        F_T(j):=\prob(E = j) = \frac{\exp(-\beta \lwt{j})}{\sum_{i=0}^{q-1}\exp(-\beta \lwt{i})} ,
    \end{equation*} 
where $\beta$ is the unique real solution to the constraint 
\begin{align*}
    T = \sum_{i=0}^{q-1} \lwt{i} \prob(E = i).
\end{align*}

We can extend this distribution for length $n$ vectors over $\Z_q$ by assuming that each coordinate is independent and identically distributed and we obtain for any $\bx \in \Z_q^n$ chosen uniformly at random:
\begin{equation*}
    F_T^{(n)}(\bx) := \prod_{i=1}^n F_T(x_i) = \frac{\exp(-\beta \lwt{\bx})}{\sum_{\by \in \Z_q^n} \exp(-\beta \lwt{\by})}.
\end{equation*}

\subsection{Lee vs. Hamming distribution}
Remember that for a given normalized Hamming weight $\delta$, we get the following marginal distribution $H_\delta(j)$ for $j\in \F_q$:
$$H_\delta(j) = \begin{cases}
        1-\delta & \mbox{if } j=0 \\
        \frac{\delta}{(q-1)} & \mbox{otherwise}
    \end{cases}.$$
Similar to the Lee metric distribution, we can extend the Hamming distribution for length $n$ vectors over $\Z_q$ by assuming that each coordinate is independent and identically distributed, i.e.,  for each $\bx \in \Z_q$
\[H_\delta^{(n)}(\bx) = \prod_{i=1}^n H_\delta(x_i) = \left(\frac{ \delta}{q-1} \right)^{\abs{\supp(\bx)}} \left(1-\delta\right)^{n - \abs{\supp(\bx)}}.\]

\begin{thm}
Let $F_T$ denote the asymptotic marginal Lee distribution and let $H_\delta$ denote the asymptotic marginal Hamming distribution from Lemma \ref{lem:marginal_Lee}, both over $\Z_q$. For any $0< T < \floor{q/2} $, let $\beta$ be the corresponding value from Equation \eqref{beta} and $c_1:=(\sum_{i=0}^{q-1} \exp(-\beta \lwt{i}))^{-1}$. Let $\delta \in (0,1)$, then 
\begin{enumerate}
    \item the R\'enyi divergence of order $\infty$ between $F_T$ and $H_\delta$ is given by:
    \begin{align*}
            R_\infty(F_T||H_\delta) & = \max \left\{ \frac{c_1}{1-\delta}, \frac{c_1e^{-\beta \nu(\beta)} (q-1)}{\delta} \right\},
        \end{align*}
        where $\nu(\beta) = 1$ if $\beta\geq 0$ and $\nu(\beta) = \floor{q/2}$ if $\beta<0$.

        \item for any given $T$, the R\'enyi divergence $R_\infty(F_T||H_\delta)$ is minimized at $\delta = \frac{e^{-\beta\nu(\beta)}(q-1)}{1+e^{-\beta\nu(\beta)}(q-1)}$, giving the lower bound:
        \begin{equation}
        R_\infty(F_T || H_\delta) \geq  c_1 + c_1 e^{-\beta\nu(\beta)} (q-1). \label{eq:Renyi_Lower}
        \end{equation}
        \item 
Assuming the coordinates are independent and identically distributed, the R\'enyi divergence of $F_T$ and $H_\delta$ for length $n$ vectors over $\Z_q$ is as follows:
\[R_\infty\left(F_T^{(n)} || H_\delta^{(n)}\right) \geq \left( c_1 + c_1 e^{-\beta\nu(\beta)} (q-1) \right)^n.\] \label{thm:Lee-Hamming-divergence}
\end{enumerate}
\end{thm}
    \begin{proof} We have
      \begin{align*}
            R_\infty(F_T||H_\delta) & := \max_{\ceil{-q/2} \leq j \leq \floor{q/2}} \frac{F_T(j)}{H_\delta(j)} \\
           & = \max \left\{ \frac{c_1}{1-\delta}, \max_{1 \leq j \leq \floor{q/2}} \frac{c_1e^{-\beta j} (q-1)}{\delta} \right\}.
        \end{align*} 
        It is easy to see that the second term is maximal at $j=1$ if $\beta \geq 0$, or $j=\floor{q/2}$ if $\beta <0$. Hence we get
        \begin{align*}
            R_\infty(F_T||H_\delta) & = \max \left\{ \frac{c_1}{1-\delta}, \frac{c_1e^{-\beta \nu(\beta)} (q-1)}{\delta} \right\},
        \end{align*}
        where $\nu(\beta) = 1$ if $\beta\geq 0$ and $\nu(\beta) = \floor{q/2}$ if $\beta<0$.
        
        Note that the first term increases as $\delta$ increases, whereas the second term decreases as $\delta$ increases. The maximum of the two terms minimizes when the first term is equal to the second term, i.e., $\frac{c_1}{1-\delta} =  \frac{c_1e^{-\beta\nu(\beta)} (q-1)}{\delta}$ or $\delta = \frac{e^{-\beta\nu(\beta)}(q-1)}{1+e^{-\beta\nu(\beta)}(q-1)}$. Thus, we get 
        \[ R_\infty(F_T||H_\delta) \geq c_1 + c_1 e^{-\beta\nu(\beta)} (q-1) .\]
 
        For vectors of length $n$, with independent and identically distributed coordinates, the inequality follows using the multiplicative property of R\'enyi divergence (Lemma \ref{lem:product_divergence}), i.e., $R_\infty\left(F_T^{(n)}||H_\delta^{(n)}\right) = \prod_{i=1}^n R_\infty(F_T||H_\delta)$.

    \end{proof}

\begin{rem}
       Plugging in the values for $\beta$ and $c_1$ for $q\in \{2,3\}$ in Equation \eqref{eq:Renyi_Lower}, we get 
    $$ R_\infty(F_T||H_\delta)\geq 1$$
    showing that the bound is tight (since the two distributions coincide) in these cases. 
\end{rem}

As $q$ increases, we can observe that the distributions $F_T$ and $H_\delta$ become quite different from each other. The following result supports this observation by showing that the R\'enyi divergence $R_\infty$ between them goes to infinity as $q$ goes to infinity. 

\begin{prop}
    Let $T \in (0,\floor{q/2}]$ and $\delta \in (0,1)$ be constant (with respect to $q$) real numbers. Then, 
    the R\'enyi divergence $R_\infty(F_T || H_\delta)$ goes to infinity as $q$ goes to infinity. 
\end{prop}
\begin{proof}
\begin{sloppypar}
From Theorem \ref{thm:Lee-Hamming-divergence}, we have that $R_\infty(F_T || H_\delta)  \geq  c_1 + c_1 e^{-\beta\nu(\beta)} (q-1), $ where $c_1 = (\sum_{i=0}^{q-1} \exp(-\beta \lwt{i}))^{-1}$
and $\beta$ is the unique real solution of Equation \eqref{beta}, i.e.,
\[T =  \sum_{j=0}^{q} \lwt{j} \frac{\exp(-x \ \lwt{j})}{\sum_{i=0}^{q} \exp(-x \ \lwt{i})}.\]

Let $M = \floor{q/2}$. Then, as $q$ goes to infinity, we can rewrite the above equation as follows:
\begin{align*}
    T&=\lim_{q \rightarrow \infty}  \frac{\sum_{i=0}^{q} \lwt{i} \exp(-x \ \lwt{i})}{\sum_{i=0}^{q} \exp(-x \ \lwt{i})} \\
    &=  \lim_{M \rightarrow \infty}
    \frac{2\sum_{i=1}^{M} i \exp(-x \ i)}{1+2\sum_{i=1}^{M} \exp(-x \ i)} \\
    &=
    \frac{2\exp(-x)}{(1-\exp(-x))^2}\cdot
    \frac{1-\exp(-x)}{\exp(-x)+1} \\
    &= 
    \frac{2\exp(x)}{\exp(2x)-1}.
     \end{align*}
Therefore, as $q$ goes to infinity, $\beta$ converges to the positive real solution of $T = 2 \exp(x)/(\exp(2x) -1)$, i.e., 
$$\exp(\beta) \to \frac{1+\sqrt{1+T^2}}{T} \ \mbox{as } q \to \infty.$$
Using this, it is easy to check that, as $q \to \infty$, $$c_1 \to \frac{\exp(\beta) - 1}{\exp(\beta) +1} = \frac{(1-T) + \sqrt{1+T^2}}{(1+T) + \sqrt{1+T^2}}.$$
As a conclusion, we note that for a fixed $T$ (constant with respect to $q$), both $e^{-\beta}$ and $c_1$ converge to a constant as $q$ tends to infinity. Hence, $R_\infty(F_T||H_\delta) \geq  c_1 + c_1 e^{-\beta} (q-1)$ diverges as $q$ goes to infinity. 
This supports the intuition that the Lee metric diverges away from the Hamming metric with growing $q$.
\end{sloppypar}
\end{proof}

\subsection{Laplace vs. Gaussian distribution}

We motivate this section by first showing the connection between the Lee metric and the discrete Laplace distribution. 
Recall that the discrete Laplace distribution over $\Z$ with parameter $b>0$ is defined as 
        \[\Lap_{\Z,b}(x) = \frac{\exp\left( - \abs{x}/b \right)}{\sum_{y \in \Z} \exp(- \abs{y}/b)}.\]

On the other hand, by considering the representation of $\Z_q$ centered at origin, the marginal Lee distribution $F_T(j)$ can be rewritten as
\begin{equation}
        F_T(j) = \frac{\exp(-\beta \abs{j})}{\sum_{i\in \Z_q}\exp(-\beta \abs{i})} , \label{eq:Lee_dist_centered}
\end{equation} 
for each $j \in \Z_q = \{-\lfloor(q-1)/2\rfloor,\ldots,\lfloor q/2\rfloor\}$. 
We can observe that the above two distributions would coincide when $b= 1/\beta$ and $q$ goes to infinity. We can deduce the same for length $n$ vectors as well, by assuming that for each coordinate the distributions are independent and identical.

\begin{lem}
Let $F_T$ be the Lee distribution over $\Z_q^n$ from Lemma \ref{lem:marginal_Lee}, and let $\Lap_{\Z^n,b}$ be the discrete Laplace distribution over $\mathbb Z^n$. Assuming that each coordinate is independent and identically distributed, we get that
$$\lim_{q\rightarrow \infty} F_T(\bx) = \Lap_{\Z^n,\frac{1}{\beta}}(\bx),$$ for every $\bx \in \Z_q^n$.
\end{lem}
\begin{proof}
From Equation \eqref{eq:Lee_dist_centered}, we have
        \begin{align*}
F_{T}(\bx)  &= \frac{1}{\left(\sum_{i \in \Z_q}\exp(-\beta \abs{i})  \right)^n } \prod_{j=1}^{n} \exp(-\beta \abs{x_j})   \\
& = \frac{1}{\sum_{\by \in \Z_q^{n}} \exp(-\beta \norm{\by}_1)} \exp(-\beta \norm{\bx}_1).
\end{align*}

Now, if we take the limit of $q$ to infinity, then we see that the above Lee distribution converges to the discrete Laplace distribution $\Lap_{\Z,b}$ with parameter $b = 1/\beta$.
\end{proof}

The above lemma says that the Lee metric distribution is close to the Laplace distribution for large $q$. Therefore, in the case of large $q$, it would make sense to compare Laplace and Gaussian distribution in order to compare Lee and Euclidean error distributions. 
In other words, such a comparison would give us a good understanding when lattice techniques in the $\ell_2$-norm are beneficial to solve $\ell_1$-norm lattice problems or Lee metric decoding problems. 

We start by comparing the continuous version of Laplace and Gaussian distribution by computing the Kullback-Leibler divergence between them.

\begin{thm}\label{thm:KL}
    The Kullback-Leibler divergence between the Laplace distribution $\Lap_{\R,b}$ and the Gaussian distribution $\D_{\R,\sigma}$ is given by
    \[KL(\Lap_{\R,b}||\D_{\R,\sigma}) = \log\left(\frac{\sigma\sqrt{\pi/2}}{b} \right) + \frac{b^2}{\sigma^2} - 1.\] 
    For vectors of length $n$, we assume the coordinates are  independent and identically distributed, and obtain:
    \[KL(\Lap_{\R^n,b}||\D_{\R^n,\sigma}) = n \left(\log\left(\frac{\sigma\sqrt{\pi/2}}{b} \right) + \frac{b^2}{\sigma^2} - 1 \right).\]
\end{thm}
\begin{proof}
 We compute the KL divergence as follows:
        \begin{align*}
            KL(\Lap_{\R,b}||\D_{\R,\sigma}) &:= \int_{-\infty}^{\infty} \Lap_{\R,b}(x) \log \left(\frac{\Lap_{\R,b}(x)}{\D_{\R,\sigma}(x)}\right) dx \\
            & = \int_{-\infty}^{\infty} \frac{1}{2b} e^{-\abs{x}/b} \left(\log\left( \frac{1}{2b} e^{-\abs{x}/b}\right) - \log\left( \frac{1}{\sigma\sqrt{2\pi}} e^{-x^2/2\sigma^2}\right) \right) dx\\
            & = \int_{-\infty}^{0} \frac{1}{2b} e^{x/b} \left(\frac{x^2}{2\sigma^2} + \frac{x}{b} + \log\left( \frac{\sigma \sqrt{\pi/2}}{b}\right) \right)  dx\\
            &+ \int_{0}^{\infty} \frac{1}{2b} e^{-x/b} \left(\frac{x^2}{2\sigma^2} - \frac{x}{b} + \log\left( \frac{\sigma \sqrt{\pi/2}}{b}\right) \right)  dx.
        \end{align*}
        Using integration by parts multiple times, we obtain
        \begin{align*}
            KL(\Lap_{\R,b}||\D_{\R,\sigma}) &= \log\left(\frac{\sigma\sqrt{\pi/2}}{b} \right) + \frac{b^2}{\sigma^2} - 1.
        \end{align*}
        For vectors of length $n$, the result follows using the additive property of KL divergence (Lemma \ref{lem:product_divergence}), i.e., $KL(\Lap_{\R^n,b}||\D_{\R^n,\sigma}) = \sum_{i=1}^n KL(\Lap_{\R,b}||\D_{\R,\sigma})$.
\end{proof}

Using standard analytical tools, we can find the minimum KL divergence between the Laplace and Gaussian distribution.
\begin{cor}
    For any given $b > 0$,
           the KL divergence $KL(\Lap_{\R,b}||\D_{\R,\sigma})$ has exactly one local minimum at $\sigma = b\sqrt{2}$. In this case, we obtain
        \[KL(\Lap_{\R,b}||\D_{\R,\sigma}) = \frac{\log(\pi) - 1}{2} \approx 0.072365 .\] 
    \label{cor:KL}
\end{cor}


The above corollary shows that for any given Laplace distribution with parameter $b$, the Gaussian distribution with variance $\sigma=b\sqrt{2}$ is the closest. Moreover, the minimal divergence value is independent of the parameter $b$. Therefore, even though $b=\frac{1}{\beta}$ when using the Laplace distribution as an approximation of the Lee distribution, and $\beta$ depends on $q$ and $\delta$, the similarity of the $\ell_2$-norm of lattice vectors and the Lee weight of codewords is generally independent of $q$ and $\delta$ (for very large $q$). 

\begin{rem}
    Theorem \ref{thm:KL} shows that the Kullback-Leibler divergence between the continuous Gaussian and the Laplacian distribution (which we take as a good approximation of the Lee distribution for large $q$) grows linearly in $n$. This indicates that as $n$ increases the Laplace distribution diverges away from the Gaussian distribution. 
\end{rem}

Recall that we took a continuous (Laplace) distribution to represent a discrete (Lee) distribution. Naturally, it would be a better approximation to take the discrete Laplace distribution as an approximation of the Lee distribution. To finalize this section, we will compute the divergence between discrete Laplace and the discrete Gaussian distribution to illustrate the behavior of them.

\begin{rem}
Note that the R\'enyi divergence between discrete Laplace $\Lap_{\Z,b}$ and discrete Gaussian $\D_{\Z,\sigma}$ is $\infty$, for any given parameters $b$ and $\sigma$. This easily follows from the following calculations: 
\begin{align*}
    R_2(\Lap_{\Z,b} | \D_{\Z,\sigma}) & = \sum_{x \in \Z} \frac{\left(\frac{1}{S_1(b)} \exp\left( - \abs{x}/b\right)\right)^2}{\left(\frac{1}{S_2(\sigma)} \exp\left( - x^2/2\sigma^2\right)\right)},
    \end{align*}
        where $S_1(b) = \sum_{y \in \Z} e^{-\abs{y}/b}$ and $S_2(\sigma) = \sum_{y \in \Z} e^{-y^2/2\sigma^2}$. This evaluates to:
    \begin{align*} 
   R_2(\Lap_{\Z,b} | \D_{\Z,\sigma}) 
    & = \frac{S_2(\sigma)}{(S_1(b))^2}  \left( 1 + 2\sum_{x \geq 1} \exp\left( \frac{x^2}{2\sigma^2} - \frac{2x}{b} \right) \right).
\end{align*}
It is easy to see that the summation goes to infinity, because $\exp\left( \frac{x^2}{2\sigma^2} - \frac{2x}{b} \right)$ goes to infinity as $x$ goes to infinity. This also implies that the R\'enyi divergence $R_a$ is infinity for all order $a \in (1,\infty]$, because $R_a(F||G)$ is non-decreasing as a function of $a$.
\end{rem}

Unlike R\'enyi divergence, the Kullback-Leibler divergence between discrete Laplace and discrete Gaussian is finite. 

\begin{thm}
    Let $b, \sigma > 0$ be real numbers. The Kullback-Leibler divergence between the discrete Laplace distribution $\Lap_{\Z,b}$ and the discrete Gaussian distribution $\D_{\Z,\sigma}$ is given by
    \[KL(\Lap_{\Z,b} | \D_{\Z,\sigma})  = \log \left( \frac{S_2(\sigma)}{S_1(b)} \right)  + \frac{1}{S_1(b)} \left( 
    \frac{e^{1/b}(e^{1/b}+1)}{(e^{1/b} -1)^3 \sigma^2} - 
    \frac{2e^{1/b}}{b(e^{1/b} -1)^2}
    \right),\]
    where $S_1(b) = \sum_{y \in \Z} e^{-\abs{y}/b}$ and $S_2(\sigma) = \sum_{y \in \Z} e^{-y^2/2\sigma^2}$. \label{thm:KL_discrete_LG}
\end{thm}
\begin{proof}
    We compute the KL divergence as follows:
    \begin{align*}
    KL(\Lap_{\Z,b} | \D_{\Z,\sigma}) & := \sum_{x \in \Z} \Lap_{\Z,b}(x) \log \left( \frac{\Lap_{\Z,b}(x)}{\D_{\Z,\sigma}(x)} \right) \\
    & = \sum_{x \in \Z} \frac{e^{-\abs{x}/b}}{S_1(b)} \left( \log \left( {\frac{e^{-\abs{x}/b}}{S_1(b)}} \right) - \log \left({\frac{ e^{-x^2/2\sigma^2}}{S_2(\sigma)}} \right) \right), 
    \end{align*}
where $S_1(b) = \sum_{y \in \Z} e^{-\abs{y}/b}$ and $S_2(\sigma) = \sum_{y \in \Z} e^{-y^2/2\sigma^2}$. Note that both $S_1(b)$ and $S_2(\sigma)$ are positive finite real numbers for any given $b, \sigma >0$. In particular, $S_1(b) = \frac{e^{1/b} + 1}{e^{1/b} - 1}$, and $S_2(\sigma)$ is an evaluation of the theta function $\theta_3$\footnote{We note that $S_2(\sigma) = \theta_3(0,e^{-1/2\sigma^2})$, where $\theta_3(u,q) = 1 + 2 \sum_{n=1}^\infty q^{n^2}\cos(2nu)$ is a theta function.}. 
    \begin{align*}
    KL(\Lap_{\Z,b} | \D_{\Z,\sigma}) & = \frac{1}{S_1(b)} \sum_{x \in \Z} e^{-\abs{x}/b} \left( \log \left( \frac{S_2(\sigma)}{S_1(b)} \right) + \frac{x^2}{2\sigma^2} - \frac{\abs{x}}{b} \right)  \\
    & = \frac{1}{S_1(b)} \left( \log \left( \frac{S_2(\sigma)}{S_1(b)} \right) + 2 \sum_{x \geq 1} e^{-x/b} \left( \frac{x^2}{2\sigma^2} - \frac{x}{b} + \log \left( \frac{S_2(\sigma)}{S_1(b)}\right) \right) \right) 
\end{align*}
To proceed forward, we apply the following identities: 
\[\sum_{x\geq 1} e^{-x/b} = \frac{1}{e^{1/b} -1}; \sum_{x\geq 1} x e^{-x/b} = \frac{e^{1/b}}{(e^{1/b} -1)^2}; \sum_{x\geq 1} x^2 e^{-x/b} = \frac{e^{1/b}(e^{1/b}+1)}{(e^{1/b} -1)^3},\] and obtain
\begin{align*}
    KL(\Lap_{\Z,b} | \D_{\Z,\sigma}) & = \log \left( \frac{S_2(\sigma)}{S_1(b)} \right)  + \frac{1}{S_1(b)} \left( 
    \frac{e^{1/b}(e^{1/b}+1)}{(e^{1/b} -1)^3 \sigma^2} - 
    \frac{2e^{1/b}}{b(e^{1/b} -1)^2}
    \right).
\end{align*}
\end{proof}

\paragraph{Lower bound on KL divergence between discrete Laplace and discrete Gaussian:}
Given a fixed $b>0$, we can numerically find $\sigma>0$ that minimizes the KL divergence between $\Lap_{\Z,b}$ and $\D_{\Z,\sigma}$. Using Theorem \ref{thm:KL_discrete_LG}, we can observe that for a fixed $b$ the KL divergence $KL(\Lap_{\Z,b} | \D_{\Z,\sigma})$ is a continuous and differentiable function of $\sigma \in (0,\infty)$. Moreover, it has exactly one minimum point  $\sigma_{\min} \in (0,\infty)$, which is the positive real root of 
\begin{equation}
\frac{1}{S_2(\sigma)} \frac{\partial S_2(\sigma)}{\partial \sigma} - \frac{2 e^{1/b}}{(e^{1/b}-1)^2\sigma^3} = 0.
\label{eq:min_sigma_KL_discrete}
\end{equation}

In Table \ref{tab:KL_discrete_DG}, we provide some numerical estimates of the minimum point $\sigma_{\min}$ and the corresponding KL divergence, for different given values of $b$. We note that the minimum KL divergence converges to $\frac{\log(\pi)-1}{2} \approx 0.072365$ as the Laplace width $b$ increases (see Figure \ref{fig:KL_discrete_DG}). Recall that the constant $\frac{\log(\pi)-1}{2}$ is the minimum KL divergence between continuous Laplace and Gaussian distributions, as seen in Corollary \ref{cor:KL}. Thus, similar to the continuous case, we can conclude that the Laplace distribution diverges from the Gaussian distribution for growing length $n$.

\begin{table}[h]
    \centering
    \begin{tabular}{c|c|c}
       Laplace width  & Minimum point  & Minimum divergence  \\ 
         $b$ & $\sigma_{\min}$ & $KL(\Lap_{\Z,b} | \D_{\Z,\sigma})$ \\
       \hline
        0.1   & 0.223609 & 7.83 x $10^{-8}$\\
        0.5   & 0.607753  & 0.0886053 \\
        1.0   & 1.35696 & 0.101332 \\
        2.0   & 2.79918 & 0.0819178 \\
        4.0   & 5.64215 & 0.0749139 \\
        8.0   & 11.3063 & 0.0730125 
    \end{tabular}
    \caption{Numerical estimates of the minimal KL divergence between discrete Laplace $\Lap_{\Z,b}$ and discrete Gaussian distribution $\D_{\Z,\sigma}$, for various given values of parameter $b$. Here, the  minimum point $\sigma_{\min}$ is a solution to Equation \eqref{eq:min_sigma_KL_discrete} that minimizes $KL(\Lap_{\Z,b} | \D_{\Z,\sigma})$ for the given $b$.}
    \label{tab:KL_discrete_DG}
\end{table}

\begin{figure}
    \centering
    \includegraphics[width=0.48\linewidth]{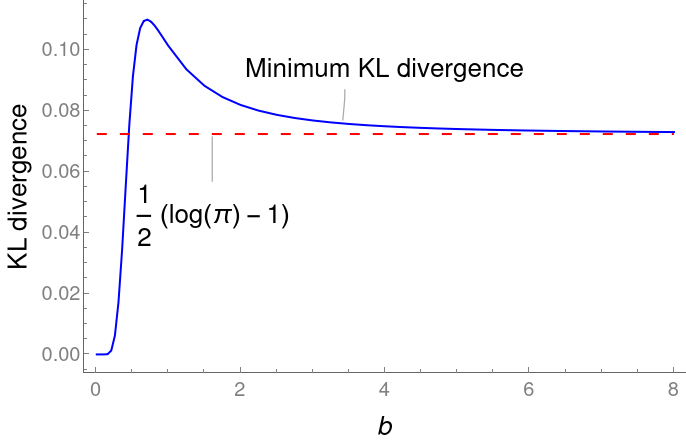} \hfill
    \includegraphics[width=0.48\linewidth]{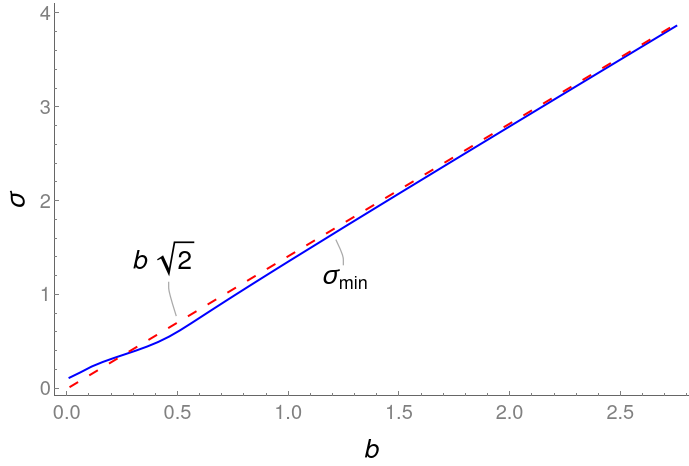}
    \caption{Numerical estimates of the minimal KL divergence between Laplace and Gaussian distribution. 
    In the left figure, we plot the minimum $KL(\Lap_{\Z,b}|\D_{\Z,\sigma})$ (solid blue line) as a function of $b$ and compare it with the constant $\frac{\log(\pi)-1}{2}$ (dashed red line) corresponding to the continuous case (Corollary \ref{cor:KL}). In the right figure, we plot the corresponding $\sigma_{\min}$ where the minimum divergence is achieved. Here again we compare $\sigma_{\min}$ with the minimum sigma $\sigma = b \sqrt{2}$ obtained in the continuous case (Corollary \ref{cor:KL}).
    }
    \label{fig:KL_discrete_DG}
\end{figure}





\section{Conclusions}

Due to the recent developments in Lee metric code-based cryptography, in particular the lattice-based attack on the NIST submission FuLeeca, we analyzed the connection of Lee metric code-based cryptosystems and the corresponding lattice problems in the $\ell_1$- and the $\ell_2$-norm.

In particular, we showed that there are polynomial time reductions in both directions between the Lee decoding problem over $\Z_q$ and the unique shortest vector problem (via the bounded distance decoding problem) over $\Z$ with respect to the $\ell_1$-norm (where for the reduction from $\LeeDP_t$ to $\BDD_\alpha$ we require that $t<q$). 
Moreover, we gave a lower bound on the number of points that are contained in the lattice generated by a given code basis, showing that this number depends on $q$ and the actual choice of basis. The bound suggests that the success likelihood of an attack by finding vectors in this lower dimensional lattice increases for growing $q$. 
Furthermore, we studied the divergence behavior of various probability distributions connected to the Lee and Hamming weight, as well as the $\ell_1$- and $\ell_2$-norm. Our results show that the behavior of the Lee metric diverges from the one of the Hamming metric for growing modulus $q$, and that the Laplace distribution diverges from the Gaussian distribution for growing vector length $n$ (for large $q$). 

These results show when (Hamming and) lattice techniques can be used to break Lee metric code-based cryptosystems. Hence, this can tell us which parameters should be avoided when designing new public key encryption schemes or digital signatures using Lee metric error correcting codes. In particular, when $q$ is chosen extremely small, then Hamming-based coding techniques can be used to attack Lee metric cryptographic schemes. On the other hand, for large $q$, $\ell_1$-lattice techniques might be applicable, however, using $\ell_2$-techniques will most likely not work. Moreover, when $q$ is very large, then the idea of the FuLeakage attack of using the lattice generated by a basis of the code (instead of the whole $\ConstructionA$ lattice) has a higher success probability than for smaller $q$. This means that in general, $q$ should be chosen large enough such that Hamming attacks are not applicable but also small enough that a large part of the $\LAG$ lattice is not contained in $\LA$. Furthermore, $n$ should be chosen large enough that the $\ell_2$-norm is not a good approximation of the $\ell_1$-norm. The exact parameters depend on the actual cryptosystem, the relationship of modulus, length, and minimum distance of the code and need to be investigated when designing a new Lee metric code-based cryptosystem.

\section*{Acknowledgements}

Carlos Vela Cabello is supported by the \emph{Grundlagenforschungsfond (GFF)} of the University of St.Gallen, project no.\ 2260780. The authors wish to thank Cameron Foreman and Kevin Milner for useful comments on the manuscriptreviewing the manuscript.

\bibliographystyle{splncs04}
\bibliography{references}

\begin{thebibliography}{10}
\providecommand{\url}[1]{\texttt{#1}}
\providecommand{\urlprefix}{URL }
\providecommand{\doi}[1]{https://doi.org/#1}

\bibitem{AjtaiD97}
Ajtai, M., Dwork, C.: A {P}ublic-{K}ey {C}ryptosystem with
  {W}orst-{C}ase/{A}verage-{C}ase {E}quivalence. In: Proceedings of the
  Twenty-Ninth Annual ACM Symposium on Theory of Computing. p. 284–293. STOC
  '97, Association for Computing Machinery, New York, NY, USA (1997).
  \doi{10.1145/258533.258604}

\bibitem{antonio2011decoding}
Antonio, C.d.A., Jorge, G.C., Costa, S.I.: Decoding q-ary lattices in the {L}ee
  metric. In: 2011 IEEE Information Theory Workshop. pp. 220--224. IEEE (2011).
  \doi{10.1109/ITW.2011.6089382}

\bibitem{bariffi2021properties}
Bariffi, J., Bartz, H., Liva, G., Rosenthal, J.: On the properties of error
  patterns in the constant {L}ee weight channel. In: International Zurich
  Seminar on Information and Communication (IZS) (2022).
  \doi{10.3929/ethz-b-000535277}

\bibitem{bariffi2022information}
Bariffi, J., Khathuria, K., Weger, V.: Information set decoding for
  {L}ee-metric codes using restricted balls. In: Code-Based Cryptography
  Workshop. pp. 110--136. Springer (2022). \doi{10.1007/978-3-031-29689-5\_7}

\bibitem{chailloux2021classical}
Chailloux, A., Debris-Alazard, T., Etinski, S.: Classical and quantum
  algorithms for generic syndrome decoding problems and applications to the
  {L}ee metric. In: Post-Quantum Cryptography: 12th International Workshop,
  PQCrypto 2021, Daejeon, South Korea, July 20--22, 2021, Proceedings 12. pp.
  44--62. Springer (2021). \doi{10.1007/978-3-030-81293-5\_3}

\bibitem{conwaylattices}
Conway, J.H., Sloane, N.J.A.: Lattices with few distances. Journal of number
  theory  \textbf{39}(1),  75--90 (1991). \doi{10.1016/0022-314X(91)90035-A}

\bibitem{R-KL_divergence}
van Erven, T., Harremos, P.: {R}ényi {D}ivergence and {K}ullback-{L}eibler
  {D}ivergence. IEEE Transactions on Information Theory  \textbf{60}(7),
  3797--3820 (2014). \doi{10.1109/TIT.2014.2320500}

\bibitem{horlemann2021information}
Horlemann-Trautmann, A.L., Weger, V.: Information set decoding in the {L}ee
  metric with applications to cryptography. Advances in Mathematics of
  Communications  \textbf{15}(4) (2021). \doi{10.3934/AMC.2020089}

\bibitem{FuLeakage}
Hörmann, F., van Woerden, W.: Fuleakage: Breaking {F}u{L}eeca by {L}earning
  {A}ttacks. In: Advances in Cryptology - {CRYPTO} 2024 Proceedings. Lecture
  Notes in Computer Science, vol. 14925, pp. 253--286. Springer (2024).
  \doi{10.1007/978-3-031-68391-6\_8}

\bibitem{lyubashevsky2009bounded}
Lyubashevsky, V., Micciancio, D.: On bounded distance decoding, unique shortest
  vectors, and the minimum distance problem. In: Advances in Cryptology -
  {CRYPTO} 2009 Proceedings. Lecture Notes in Computer Science, vol.~5677, pp.
  577--594. Springer. \doi{10.1007/978-3-642-03356-8\_34}

\bibitem{LyubashevskyPR13}
Lyubashevsky, V., Peikert, C., Regev, O.: On {I}deal {L}attices and {L}earning
  with {E}rrors over {R}ings. J. {ACM}  \textbf{60}(6),  43:1--43:35 (2013).
  \doi{10.1145/2535925}

\bibitem{mink-ucsd}
Micciancio, D.: Lecture notes on {I}ntroduction to {L}attices,
  \url{https://cseweb.ucsd.edu/classes/wi12/cse206A-a/lec1.pdf}

\bibitem{Regev09}
Regev, O.: On lattices, learning with errors, random linear codes, and
  cryptography. J. {ACM}  \textbf{56}(6),  34:1--34:40 (2009).
  \doi{10.1145/1568318.1568324}, preliminary version in STOC 2005.

\bibitem{ritterhoff2023fuleeca}
Ritterhoff, S., Maringer, G., Bitzer, S., Weger, V., Karl, P., Schamberger, T.,
  Schupp, J., Wachter{-}Zeh, A.: Fu{L}eeca: {A} {L}ee-{B}ased {S}ignature
  {S}cheme. In: Code-Based Cryptography - 11th International Workshop, CBCrypto
  2023. Lecture Notes in Computer Science, vol. 14311, pp. 56--83. Springer
  (2023). \doi{10.1007/978-3-031-46495-9\_4}

\bibitem{RushSloane}
Rush, J.A., Sloane, N.J.A.: An improvement to the {M}inkowski-{H}iawka bound
  for packing superballs. Mathematika  \textbf{34}(1),  8--18 (1987).
  \doi{https://doi.org/10.1112/S0025579300013231}

\bibitem{SS10}
Sakzad, A., Sadeghi, M.: On cycle-free lattices with high rate label codes.
  Adv. Math. Commun.  \textbf{4}(4),  441--452 (2010).
  \doi{10.3934/AMC.2010.4.441}

\bibitem{mink-epfl}
Shmonin, G.: Lecture notes on {M}inkowski’s theorem and its applications,
  \url{https://www.epfl.ch/labs/disopt/wp-content/uploads/2018/09/minkowski.pdf}

\bibitem{Vaa79}
Vaaler, J.: A geometric inequality with applications to linear forms. Pacific
  Journal of Mathematics  \textbf{83}(2),  543--553 (1979).
  \doi{10.2140/pjm.1979.83.543}

\bibitem{weger2022hardness}
Weger, V., Khathuria, K., Horlemann, A.L., Battaglioni, M., Santini, P.,
  Persichetti, E.: On the hardness of the {L}ee syndrome decoding problem.
  Advances in Mathematics of Communications  \textbf{18}(1),  233--266 (2024).
  \doi{10.3934/amc.2022029}

\end{thebibliography}

\appendix

\section{Proof of Minkowski's convex body theorem (Theorem \ref{thm:minkowski})} \label{app:Minkowski}

In order to prove Theorem \ref{thm:minkowski}, we use Blichfeldt's theorem on non-full dimensional lattices. In the following, our proofs are based on the proofs from \cite[Theorem 20-21]{mink-ucsd} and \cite[Theorem 5-6]{mink-epfl}. 

\begin{thm}[Blichfeldt]
    Let $\L$ be a $k$-dimensional lattice in $\R^n$ and $S \subseteq \Span_\R(\L)$ be a convex set symmetric about the origin (i.e., $\bx \in S$ implies $-\bx \in S$). Suppose that $\Vol_k(S)>m\cdot \det(\L)$, for some integer $m$. Then, there are $m+1$ vectors $\bz_1,\ldots,\bz_{m+1}$ in $S$ such that $\bz_i - \bz_j \in \L$ for each $i,j$. 
    \label{thm:Blichfeldt_OG}
\end{thm}
\begin{proof}
    Let $\bB = \{\bb_1, \ldots,\bb_k\}$ be a basis of $\L$ and $\Pi(\bB)$ be the fundamental parallelepiped associated to $\bB$ defined as $\Pi(\bB):= \left\{\sum_{i=1}^k x_i\bb_i \mid x_i \in [0,1)\right\}$.  Consider the sets $S_\bx := S \cap \{\by + \bx \mid \by \in \Pi(\bB) \}$ for each $\bx \in \L$. We note that these sets form a partition of $S$, i.e., they are pairwise disjoint and $S = \bigcup_{\bx \in \L} S_{\bx}$. Thus, we have $\Vol_k(S) = \sum_{\bx \in \L} \Vol_k(S_\bx)$.
    
    Now consider the shifted sets $S_\bx - \bx := \{\by - \bx \mid \by \in S_\bx\}$. We note that $S_\bx - \bx = (S - \bx) \cap \Pi(\bB).$ Now, since $\Vol_k(S_\bx) = \Vol_k(S_\bx - \bx)$, we have that 
    \[\sum_{\bx \in \L} \Vol_k(S_\bx-\bx) = \sum_{\bx \in \L} \Vol_k(S_\bx) = \Vol_k(S) > m \cdot \determ(\L) = m \cdot \Vol(\Pi(\bB)).\]
    From $\sum_{\bx \in \L} \Vol_k(S_\bx-\bx) > m \cdot \Vol(\Pi(\bB))$ and $S_\bx - \bx \subseteq \Pi(\bB)$, we deduce that there exist $m+1$ distinct points $\bx_1,\ldots,\bx_{m+1} \in \L$  such that $\bigcap_{i=1}^{m+1} (S_{\bx_i}-\bx_i)$ is non-empty.  Let $\by \in \bigcap_{i=1}^{m+1} (S_{\bx_i}-\bx_i)$ and $\bz_i = \by + \bx_i \in S_{\bx_i} \subseteq S$ for each $i \in \{1,\ldots,m+1\}$. Thus, we have $m+1$ vectors $\bz_1,\ldots,\bz_{m+1} \in S$ such that $\bz_i - \bz_j = \bx_i - \bx_j \in \L$ for each $i,j$.
\end{proof}

Now, we prove Theorem \ref{thm:minkowski} as a corollary to Blichfeldt's theorem. 

\begin{thm}[Theorem \ref{thm:minkowski}]
Let $\L$ be a $k$-dimensional lattice in $\R^n$ and let $S \subseteq \Span_\R(\L)$ be a convex set symmetric about the origin (i.e., $\bx \in S$ implies $-\bx \in S$).
Suppose that $\Vol_k(S) > m \cdot 2^k \cdot \determ(\L)$, for some integer $m$.
Then there are $m$ different pairs of vectors $\pm \bz_1, \ldots, \pm \bz_m \in S \cap \L \setminus \{0\}$.
\end{thm}
\begin{proof}
    Consider the set $\frac{1}{2}S = \{\bx \mid 2\bx \in S\}$, then it is easy to note that $\Vol_k\left(\frac{1}{2}S\right) = \frac{1}{2^k}\Vol_k(S)$. This follows since $S \subseteq \Span_\R(\L)$ is contained in a $k$-dimensional subspace of $\R^n$ and we can apply an orthogonal transformation to embed $S$ in $\R^k$ without changing its volume.

    Now, since $\Vol_k\left(\frac{1}{2}S\right) > m \cdot \det(\L)$, we apply Theorem \ref{thm:Blichfeldt_OG} to obtain $m+1$ vectors $\frac{1}{2}\bx_1,\ldots,\frac{1}{2}\bx_{m+1} \in \frac{1}{2}S$ such that  $\frac{1}{2}\bx_i - \frac{1}{2}\bx_j \in \L$ for all $i,j$. We assume that $\bx_1$ is the smallest vector with respect to the lexicographic order $\prec$. 
    
    Define $\bz_i = \frac{1}{2}\bx_{i+1} - \frac{1}{2} \bx_1 \in \L$ for each $i \in \{1,\ldots,m\}$. Clearly, $\bz_i$'s are distinct vectors, and since $0 \prec \bz_i$ for all $i$, we have $\bz_i \neq - \bz_j$ for all $i,j.$ Finally, since $S$ is convex and symmetric, $\bz_i = \frac{1}{2}\bx_{i+1} + \frac{1}{2} (-\bx_1) \in S$ for all $i$.
    
\end{proof}

\end{document}